\documentclass[runningheads]{llncs}

\usepackage{graphicx}
\usepackage[dvipsnames]{xcolor}
\usepackage[colorlinks=true, allcolors=blue,linkcolor=blue, citecolor=blue,backref=page]{hyperref}
\usepackage{listings}
\usepackage{amssymb}
\usepackage{amsfonts}
\usepackage{mathtools}
\usepackage{maple}
\usepackage[utf8]{inputenc}
\usepackage{amsmath}
\usepackage{breqn}
\usepackage{orcidlink}
\usepackage{mathrsfs}
\usepackage{listings}
\usepackage{microtype}
\usepackage{moreverb}
\usepackage{verbatim}
\usepackage{sverb}
\usepackage{geometry}
\usepackage{color}
\usepackage{makecell}
\usepackage{url}
\usepackage{cleveref}
\usepackage{textcomp}
\usepackage{moreverb}
\usepackage{verbatim}
\usepackage{sverb}
\usepackage{algorithm}
\usepackage{algpseudocode}
\usepackage{algorithmicx}
\usepackage{euscript}

\newcommand*{\QEDA}{\hfill\ensuremath{\blacksquare}}

\newcommand*{\CQFD}{\hfill\ensuremath{\square}}

\newcommand{\NN}{\mathbb{N}}

\newcommand{\KK}{\mathbb{K}}

\begin{document}
    \title{Arithmetic of D-Algebraic Functions}

\author{Bertrand Teguia Tabuguia\orcidlink{0000-0001-9199-7077}}

\authorrunning{B. Teguia Tabuguia}

\institute{Max Planck Institute for Mathematics in the Sciences, Leipzig, Germany\\
           Department of Computer Science, University of Oxford, United Kingdom\\
    \email{bertrand.teguia@cs.ox.ac.uk}}

\maketitle              

\begin{abstract}
    We are concerned with the arithmetic of solutions to ordinary or partial nonlinear differential equations which are algebraic in the indeterminates and their derivatives. We call these solutions D-algebraic functions, and their equations are algebraic (ordinary or partial) differential equations (ADEs). The general purpose is to find ADEs whose solutions contain specified rational expressions of solutions to given ADEs. For univariate D-algebraic functions, we show how to derive an ADE of smallest possible order. In the multivariate case, we introduce a general algorithm for these computations and derive conclusions on the order bound of the resulting algebraic PDE. Using our accompanying Maple software, we discuss applications in physics, statistics, and symbolic integration.

\end{abstract}

\section{Introduction}

Let $\mathbb{K}$ be a field of characteristic zero, and $R\coloneqq \left(\mathbb{K}(x),\frac{d}{dx}\right)$ a differential field. We denote by $S_y~\coloneqq~R[y,y',\ldots]\coloneqq R[y^{(\infty)}],$ $\frac{d^j\, y}{dx}=y^{(j)}$, the differential polynomial ring of the differential indeterminate $y$. Elements of that ring are called differential polynomials. Differential ideals are ideals that are closed under taking derivatives. 
\begin{definition}\label{def:def1}
A function $f\coloneqq f(x)$ is called differentially algebraic, or simply D-algebraic, in $S_y$ if there exists a differential polynomial $p\in S_y$ such that $p(f)=0$, i.e., $p(y)\mid_{y=f(x)}=0$. 
\end{definition}
Therefore D-algebraic functions can be seen as zeros of differential polynomials. We will often use this terminology. What we call functions are always elements of some $\KK(x)$-algebra which is closed under the derivation used; they could be analytic functions, formal power series, or distributions in the appropriate domains. They are mainly differentiable objects that are suitable for multiplications by the indeterminates.

\begin{example}[Univariate D-algebraic]\label{ex:ex1} The Weierstrass $\wp\coloneqq\wp(x,g_2,g_3)$ function is D-algebraic since it is a zero of ${y'}^3-4\,y^3+g_2\,y+g_3\in S_y$, with $\mathbb{K}=\mathbb{Q}(g_2,g_3)$. \QEDA
\end{example}
The equation resulting from equating a differential polynomial to zero is called an algebraic differential equation (ADE). For a given differential polynomial, such an equation will be called associated ADE. Similarly, given an ADE, the differential polynomial obtained by removing the equality to zero is called its associated differential polynomial. The order of a differential polynomial is the one of its associated ADE, and its total degree, which we simply call degree, is the total degree of the differential polynomial seen as a multivariate polynomial over $\KK(x)$.

\Cref{def:def1} generalizes to the multivariate case with the partial differential ring $R=\left(\KK(\mathbf{x}),\Delta_n\right)$, $\mathbf{x}=x_1,\ldots,x_n$, and $\Delta_n=\lbrace \partial_{x_i}\coloneqq \frac{\partial}{\partial x_i}, 1\leq i \leq n\rbrace$, $n\in\NN$. In that case, a D-algebraic function is seen as a zero of a partial differential polynomial in $R[y^{(\infty,\overset{n}{\ldots},\infty)}]$, where $(\infty,\overset{n}{\ldots},\infty)$ is a tuple with $n$ components of $\infty$. 
\begin{example}[Multivariate D-algebraic]\label{ex:ex2} The so-called harmonic functions are D-algebraic since they are zeros of the partial differential polynomial
\begin{equation}\label{eq:eq1}
    y^{(2,0,0)} \, + \, y^{(0,2,0)} \, + \, y^{(0,0,2)},
\end{equation}
associated to the Laplace equation $\frac{\partial^2y}{\partial x_1^2} + \frac{\partial^2y}{\partial x_2^2} + \frac{\partial^2y}{\partial x_3^2}=0$.\QEDA
\end{example}

Note that the notion of D-algebraicity used here was already defined in \cite[Chapter IV]{kolchin1973differential}; it relaxes the definitions given in \cite{vdH19,raschel2020counting} as it does not require D-algebraicity in each independent variable. We will manipulate D-algebraic functions as generic zeros \cite[Chapter IV.2]{kolchin1973differential} of the general component of some differential ideal defined from the ADEs under consideration. An essential ingredient for the correctness of our algorithms is the notion of triangular sets, which is discussed in detail in \cite{hubert2003notes}.

The set of D-algebraic functions is perhaps the most general structural class of functions that encompasses most functions encountered in the sciences. One prominent area where they naturally occur is enumerative combinatorics where D-algebraic functions are regarded as generating functions of counting problems. Recently in \cite{raschel2020counting}, Bernardi, Bousquet-M\'{e}lou, and Raschel classified some D-algebraic quadrant models that are not differentially finite (D-finite), i.e., the generating functions of the corresponding counting problems are D-algebraic but their derivatives do not span a finite dimensional vector space (see \cite{stanley1980differentiably,kauers2015walks}). This motivates to explore beyond D-finiteness, a structure that interacts favorably with linear algebra techniques, unlike D-algebraic functions. It is thus important to study the effectiveness of the closure properties of D-algebraic functions. These functions form a field, and are closed under many other operations like composition and differentiation \cite{RSB2023} \cite[Section 6]{raschel2020counting} \cite{moore1896concerning}. The theory behind D-algebraicity began earlier in the last century, see for instance \cite{moore1896concerning,maillet1903series,ritt1950differential}, \cite[Chapter II]{kolchin1973differential}, \cite{denef1984power}, and \cite{lipshitz1986gap}. However, these studies primarily focus on theoretical aspects regarding operations that may or may not preserve D-algebraicity, rather than their practical effectiveness. One had to wait until the end of the last century, when algorithms like the Rosendfeld-Gr\"{o}bner algorithm for representing radical differential ideals appeared \cite{boulier1995representation}. This algorithm is implemented in the \texttt{DifferentialAlgebra} package \cite{DAMaplepackage} in Maple. Of a similar flavor is the Thomas decomposition algorithm which views the generators of radical differential ideals as the ``simple systems'' of a system of algebraic PDEs \cite{bachler2012algorithmic} \cite[Chapter II]{robertz2014formal}. It is available in the \texttt{DifferentialThomas} package of Maple. 

Both algorithms may serve as alternatives to our computations in certain contexts. We also mention the connection between ADEs and biochemical reaction networks which arise from particular types of dynamical systems \cite{hong2020global,dong2023differential}. As we will see later, the theory developed in \cite{hong2020global} is essential to our algorithmic approach. Other interesting topics related to D-algebraic functions are that of differentially definable functions \cite{jimenez2020some}, and solutions to ADEs of fixed degrees \cite{TeguiaDelta2}. We are confident that our results will enable certain progress around the study of these D-algebraic sub-classes.  There is a method for constructing power series that are not D-algebraic, as discussed in \cite{moore1896concerning} and \cite{lipshitz1986gap}. The idea is based on gaps between non-zeros coefficients of the series. One well-known example of such power series is $\sum_{n=0}^{\infty} z^{2^n}$ \cite{maillet1903series}. While these specific series are not the focus of our paper, it is worth noting that a similar strategy, as outlined in \cite{TBguessing}, adapted to higher degree ADEs in conjunction with our findings, may assist in deciding the D-algebraicity of series like $\sum_{n=0}^{\infty}z^{n^3}$  \cite{lipshitz1986gap,lipshRubel1991}.

To our knowledge, this paper and the work in \cite{RSB2023} are the only ones that are entirely dedicated to effective computations of the closure properties of D-algebraic functions. With regards to power series solutions, effective results appeared in \cite{vdH19}. Like in \cite{RSB2023}, compared to that work of van der Hoeven, which for the zero-equivalence problem requires investigations on the initial conditions, in our case we focus on what he called ``global computations'', which essentially rely on the framework of differential algebra and computational commutative algebra. The reason is that instead of dealing with a particular solution of an ADE, we deal with \textit{all} solutions of that ADE for which the arithmetic operation performed is meaningful. Thus zero-testing and initial conditions are neglected, though our results may apply to them.

This paper begins with the arithmetic of univariate D-algebraic functions. Based on the construction of systems of ordinary differential equations (ODEs), we show how to compute an ADE of order at most the sum of the orders of \textit{any} given ADEs, satisfied by some rational expression of the solutions of those given ADEs. In \Cref{sec:multivar}, we develop an algorithm for the arithmetic of multivariate D-algebraic functions and prove that the order of the resulting algebraic PDEs is not always bounded by the sum of the orders of the given algebraic PDEs. \Cref{sec:disc} presents applications using our Maple implementation of \Cref{algo:Algo1} and \Cref{algo:Algo2}. We mention that all our results also work for differential rational expressions of D-algebraic functions due to the natural action of differentiation for D-algebraic functions.

\subsection*{Some background and terminology}

Let us now assume the univariate setting. To define our target ODE system, we use the multivariate differential polynomial ring $S_{\mathbf{y},z}$, where $\mathbf{y}=(y_1,\ldots,y_n)$ for some $n\in\mathbb{N}$. Our main interest is in systems of the form
\begin{equation*}
    \begin{cases}
     \mathbf{y}' = \mathbf{A}(\mathbf{y})\\
     z = B(\mathbf{y})
    \end{cases}, \,\quad\, \left(\mathcal{M}\right)
\end{equation*}
where $\mathbf{A}=\left(A_1,\ldots,A_n\right)\in\KK(x)(y_1,\ldots,y_n)^n,$ $B\in\KK(x)(y_1,\ldots,y_n)$. In the context of differential algebra, we consider the following $n+1$ polynomials
\begin{equation}\label{eq:eq2}
    Q\mathbf{y}'-\mathbf{a}(\mathbf{y}),\, Qz-b(\mathbf{y});
\end{equation}
where $Q$ is the least common multiple of the denominators in $(\mathcal{M})$, and $A_i=a_i/Q,$ $i=1,\ldots,n, B=b/Q$. With an appropriate monomial ordering over the differential polynomial ring $S_{\mathbf{y},z}$, multiplying by $Q$ in \eqref{eq:eq2} is equivalent to clearing denominators and considering $Q$ as the least common multiple of the initials of the resulting differential polynomials. We make the above choice for convenience as it allows us to defer the selection of a monomial ordering to a later stage, which also appears more suitable for the algorithmic part (see also \cite{dong2023differential}). Next, we consider the differential ideal
\begin{equation}\label{eq:eq3}
    I_{\mathcal{M}}\coloneqq \langle Q\mathbf{y}'-\mathbf{a}(y),\, Qz-b(\mathbf{y}) \rangle \colon Q^{\infty} \subset R[\mathbf{y}^{(\infty)},z^{(\infty)}]=:S_{\mathbf{y},z},
\end{equation}
where $:Q^{\infty}$ denotes the saturation with $\{Q\}$. We define $I_{\mathcal{M}}^{\leq j}$ as the algebraic ideal containing the polynomials in \eqref{eq:eq2} and their first derivatives of order at most $j$ saturated with $\{Q\}$. The system $(\mathcal{M})$ may be seen as a \textit{single-output state-space} model without input. General setting and computations of such models are discussed in \cite{hong2020global,Glebnotes}. There it is also explained how so-called input-output equations of $(\mathcal{M})$, i.e. ADEs that relate the output variable $z$ and the input variable (usually denoted $u$), are deduced from the elimination ideal 
$$I_{\mathcal{M}}^{\leq n} \cap \KK[x][z^{(\infty)}].$$
Hence in our setting, the input-output equations correspond to ADEs satisfied by the rational function $B(\mathbf{y})$.
\begin{remark} Recall that the separant of a differential polynomial is its derivative with respect to its highest order term. Thus for all positive integers $j$, the separant of $(Q\,y_i'-a_i(\mathbf{y}))^{(j)}$ is $Q$, for $i=1,\ldots,n$. Therefore saturation by $Q$ suffices to target the generic solutions of $\left(\mathcal{M}\right)$.
\end{remark}
A system of the type of $(\mathcal{M})$ is called rational dynamical system of dimension $n$ (the dimension of $\mathbf{y}$). In their work \cite{pavlov2022realizing}, Pavlov and Pogudin developed an algorithmic approach to derive such systems from first-order ADEs of two dependent variables, which represent the output and the input. The relationship between these systems and the closure properties of D-algebraic functions is discussed in \cite{RSB2023}. Specifically, when the dynamical system $(\mathcal{M})$ models ADEs fulfilled by some D-algebraic functions $f_1,f_2,\ldots,f_N$, such that $B$ in $(\mathcal{M})$ represents a rational expression in the dependent and independent variables, the corresponding input-output equations are ADEs satisfied by $B(f_1,\ldots,f_N)$. In particular, thanks to this construction and the non-trivial nature of the elimination ideal that generates the input-output equations of $(\mathcal{M})$, we can reestablish that the set of D-algebraic functions is a field.
\begin{definition}
A differential polynomial $p\in R[y^{(\infty)}]$ of order $n$ is called {\em linear in its highest order term} (l.h.o.) if, in the expression of $p$, $y^{(n)}$ appears only linearly.
\end{definition}
\begin{example}
The differential polynomials $y''-{y'}^3+xy+7x$ and ${y''}{y'}^5y+3y^2$ are l.h.o., whereas $q={y^{(4)}}^3+xy'+1$ is not l.h.o. But its derivative $q'=3y^{(5)}{y^{(4)}}^2+y+xy''$ is l.h.o. In general, the derivative of a differential polynomial is always l.h.o.\QEDA
\end{example}

We will say that a computed algebraic (partial or ordinary) differential equation is of the least order if its order is the minimal order differential polynomial in the corresponding differential ideal. The arithmetic of two D-algebraic functions in the l.h.o. and non-l.h.o. cases were treated in \cite{RSB2023}. In this paper, we improve the computational bound given in that paper for non-l.h.o. algebraic ODEs. We demonstrate that the bound established for l.h.o. ADEs also applies to non-l.h.o. ADEs. In essence, the sum of the orders of some given ADEs serves as an upper bound for the minimal order of an ADE satisfied by rational expressions formed from generic solutions of those ADEs.

\section{Modeling the arithmetic of univariate D-algebraic functions}

Throughout this section, we consider $N+1$ D-algebraic functions $f_1,\ldots,f_N$, and $f=r(f_1,\ldots,f_N)$, where $r\in\KK(x)(v_1,\ldots,v_N)$. We assume that $f_i$ is a zero of the differential polynomial $p_i\in S_{y_i}$, $i=1,\ldots,N\in\NN$ of order $n_i\in\NN$. The main goal of this section is to complement the work in \cite{RSB2023}. As we will see in \Cref{eg:eg1}, using the implementation of \Cref{algo:Algo1}, one can compute least-order ADEs in all situations.

\subsection{The l.h.o. case}\label{sec:lho}

We here recall the result from \cite{RSB2023} for the arithmetic of D-algebraic functions that are zeros of l.h.o. differential polynomials. Since each $p_i$ is l.h.o, we can write its highest derivative of the dependent variable in terms of the lower ones from the associated ADE. We write
\begin{equation}\label{eq:eq4}
    y_i^{(n_i)} = r_{p_i}(y_i,y_i',\ldots,y_i^{(n_i-1)}).
\end{equation}
Let us now construct the ODE system (here a dynamical system) that models the operation $r(f_1,\ldots,f_N)=:f$.

Define the indeterminates
\begin{equation}\label{eq:eq5}
\begin{split}
    &v_1=y_1,v_2=y_1',\ldots, v_{n_1+1}=y_1^{(n_1)},\\
    &v_{n_1+2}=y_2,v_{n_1+3}=y_2',\ldots,v_{n_1+n_2+2}=y_2^{(n_2)},\\
    &\vdots\\
    &v_{N+\sum_{i=1}^{N}n_i}=y_N^{(n_N)}.
\end{split}
\end{equation}
This implies that
\begin{equation}\label{eq:eq6}
    v_j'=v_{j+1}, \forall j\in \left[1\,,\, N+\sum_{i=1}^{N}n_i\right]\cap \NN \setminus \left\lbrace k+\sum_{i=1}^k n_i, k=1,\ldots,N \right\rbrace.
\end{equation}
Furthermore, let
\begin{equation}\label{eq:eq7}
\begin{split}
    &w_j=v_j,\,\, \text{ for }\,\, 0< j \leq n_1\\
    &w_j=v_{j+1},\,\, \text{ for }\,\, n_1< j \leq n_1+n_2\\
    &\vdots\\
    &w_j=v_{j+N},\,\, \text{ for }\,\, \sum_{i=1}^{N-1}n_i< j \leq \sum_{i=1}^{N}n_i.
\end{split}
\end{equation}
Set $M\coloneqq \sum_{i=1}^{N}n_i$. The following $M$ dimensional dynamical system models $f$:
\begin{equation*}
\begin{cases}
    &w_1'=w_2,\,\ldots,w_{n_1}'=r_{p_1}(w_1,\ldots,w_{n_1})\coloneqq A_1(w_1,\ldots,w_M)\\
    &w_{n_1+1}'=w_{n_1+2},\,\ldots,w_{n_1+n_2}'=r_{p_2}(w_{n_1+1},\ldots,w_{n_1+n_2})\coloneqq A_2(w_1,\ldots,w_M)\\
    &\vdots\\
    &w_{M}'\coloneqq A_N(w_1,\ldots,w_M)\\
    &z = r(w_1,w_{n_1+1},\ldots,w_{\sum_{i=1}^{N-1}n_i+1})
\end{cases}. \,\quad\,\quad\, (\mathcal{M}_f)
\end{equation*}
We denote the corresponding differential ideal by $I_{\mathcal{M}_f}\in S_{\mathbf{w},z}$. The above construction sustains the following proposition, which extends \cite[Proposition 4.9.]{RSB2023} for arbitrarily many differential polynomials.
\begin{proposition}\label{prop:prop1} We use the above notation. There is an algorithm that computes an ADE of order at most $M\coloneqq n_1+n_2+\ldots+n_N$, fulfilled by the D-algebraic function $f\coloneqq r(f_1,\ldots,f_N)$.
\end{proposition}
\begin{proof} Considering a lexicographic monomial ordering that rank the differential indeterminate $z$ higher than the other indeterminates, the proof relies on the fact that the elimination ideal
\begin{equation}\label{eq:eq8}
    I_{\mathcal{M}_f}^{\leq M} \cap \KK[x][z^{(\infty)}] \subset \KK[x,z^{\leq M}]
\end{equation}
is not trivial, and contains the least-order differential polynomial that has $f=r(f_1,\ldots,f_N)$ as a zero (see \cite[Corollary 3.21]{hong2020global} and \cite[Prop. 1.27]{Glebnotes}). Computationally, among the generators (Gr\"obner basis) of the elimination ideal, we select a differential polynomial of the lowest degree among those of the smallest order.\CQFD
\end{proof} 

\subsection{The non-l.h.o. case}

Suppose there are $N_1\in\NN\setminus \{0\}$  differential polynomials that are non-l.h.o. among the $p_i$'s, $N_1\leq N$. In \cite{RSB2023} this situation is overcome by replacing those differential polynomials with their first derivatives. This brings the problem back to the l.h.o. case, but increase the bound for the order of the differential equation sought by $N_1$. The main issue is that without taking derivatives, we do not easily relate the arithmetic in this case to a rational dynamical system.

Our main observation is that, provided some algorithmic changes and a little theoretical adaptation, the arguments over which the computations from the dynamical systems rely are still valid. Indeed, suppose that $p_j$, $1\leq j \leq N$ is a non-l.h.o. differential polynomial, and let $m_j$ be the degree of $y_j^{(n_j)}$ in $p_j$. We can rewrite $p_j$ as
\begin{equation}\label{eq:eqpjmj}
    c_{m_j}\,{y_j^{(n_j)}}^{m_j} + p_{j,1}(x,y_j^{(\leq n_j)}),
\end{equation}
where $c_{m_j}\in L\coloneqq \KK[x,y_j^{(\leq n-1)}]$, $p_{j,1}\in L[y_j^{(n_j)}]$ such that $\deg_{y_j^{(n_j)}}(p_{j,1})< m_j$. The coefficient $c_{m_j}$ is often called the initial of $p_j$. We make the same change of differential indeterminates as in \Cref{sec:lho} and build an ODE system of the form
\begin{equation*}
\begin{cases}
    &w_1'=w_2,\,\ldots,{w_{n_1}'}^{m_1}=A_1(w_1,\ldots,w_M,w_{n_1}')\\
    &w_{n_1+1}'=w_{n_1+2},\,\ldots,{w_{n_1+n_2}'}^{m_2}=A_2(w_1,\ldots,w_M,w_{n_1+n_2}')\\
    &\vdots\\
    &{w_{M}'}^{m_N}= A_N(w_1,\ldots,w_M,w_M')\\
    &z = r(w_1,w_{n_1+1},\ldots,w_{\sum_{i=1}^{N-1}n_i+1})
\end{cases}, \quad\,\,\quad\,\, (\mathcal{M}_f^*)
\end{equation*}
where $m_j=\deg_{y_j^{(n_j)}}(p_j),j=1,\ldots,N$. The rational functions $A_j,j=1,\ldots,N$ result from solving $p_j=0$ for ${y_j^{(n_j)}}^{m_j}$, i.e., with the same variable names in the expressions of $A_j$ and $p_j$, we can write
$$
A_j = - \frac{p_{j,1}}{c_{m_j}}.
$$
Observe that the denominators are free of the indeterminate representing $y_j^{(n_j)}$. This allows us to define $Q$ in $(\mathcal{M}_f*)$ as in $(\mathcal{M})$ by using all the denominators of the system. To  ease the theoretical understanding, let us simplify the notation by rewriting $(\mathcal{M}_f^*)$ in the following abstract form
\begin{equation}\label{eq:radratdyn}
\begin{cases}
    {\mathbf{w}'}^{\mu} = \mathbf{\mathcal{A}}(\mathbf{w}) + \mathbf{\mathcal{E}}_{\mathbf{\mathcal{A}}}(\mathbf{w}')\\
    z = B(\mathbf{w})
\end{cases}\coloneqq
\begin{cases}
    &{w_1'}^{\mu_1}=\mathcal{A}_1(w_1,\ldots,w_M) + \mathcal{E}_{\mathcal{A}_1}(w_1')\\
    &\vdots\\
    &{w_{M}'}^{\mu_M}=\mathcal{A}_M(w_1,\ldots,w_M)+\mathcal{E}_{\mathcal{A}_M}(w_M')\\
    &z = B(w_1,\ldots,w_M)
\end{cases},
\end{equation}
where 
\begin{itemize}
    \item $\mu_i,i=1,\ldots,M$, is either $1$ or one of the $m_j, j=1,\ldots,N$ in $(\mathcal{M}_f^*)$; 
    \item $\mathcal{A}_i$ (with numerator $a_i$) is either $w_{i+1}$ or the part of $A_j$, that is free of $w_j'$ in $(\mathcal{M}_f^*)$, $j=1,\ldots,N$;
    \item $\mathcal{E}_{\mathcal{A}_i}(w_i')$ (with numerator $e_{a_i}$) is either $0$ or the part of $A_j$ that contains $w_j'$ (in its numerator) in $(\mathcal{M}_f^*)$, $j=1,\ldots,N$;
    \item $B(w_1,\ldots,w_M)=r(w_1,w_{n_1+1},\ldots,w_{\sum_{i=1}^{N-1}n_i+1})$ with numerator $b$.
\end{itemize}
We call a system of the form of \eqref{eq:radratdyn} a \textit{radical-rational dynamical system}. Note that the shape of $(\mathcal{M}_f^*)$ imposes certain constraints on the radical-rational dynamical systems we use. In the context of differential algebra, we consider the differential ideal
\begin{equation}\label{eq:eq9}
    I_{\mathcal{M}_f^*}\coloneqq \langle Q{\mathbf{w}'}^{\mu}-\mathbf{a}(\mathbf{w}) - \mathbf{e}_{\mathbf{a}}(\mathbf{w}'),\, Q\,z-b(\mathbf{w}) \rangle \colon H^{\infty} \subset R[\mathbf{y}^{(\infty)},z^{(\infty)}]=:S_{\mathbf{w},z},
\end{equation}
where $H$ is the least common multiple of $Q$ and the separants of $Q{\mathbf{w}'}^{\mu}-\mathbf{a}(\mathbf{w}) - \mathbf{e}_{\mathbf{a}}(\mathbf{w}')$. These separants are explicitely given by 
$$\mu Q {\mathbf{w}'}^{\mu - 1} - \frac{\partial}{\partial \mathbf{w}'}\mathbf{e}_{\mathbf{a}}(\mathbf{w}'),$$
where derivations are taken component wise.

Let us now present our algorithmic steps for computing ADEs satisfied by rational expressions of univariate D-algebraic functions.
%We now give the Lemma that ensures the result of our algorithmic computations.

\begin{algorithm}[ht]\caption{Arithmetic of univariate D-algebraic functions }\label{algo:Algo1}
    \begin{algorithmic} 
    \\ \Require $N$ ADEs associated to the differential polynomials $p_j$ of order $n_j$ and dependent variable $y_j(x)$, $j=1,\ldots,N$; and a rational function $r \in\KK(x)(v_1,\ldots,v_n)$.
    \Ensure An ADE of order at most $M\coloneqq n_1+\cdots+n_N$ fulfilled by the D-algebraic function $f=$~$r(f_1,\ldots,f_N)$, where each $f_j$ is a zero of the generic component of $p_j$.
    \begin{enumerate}
    \item Construct $(\mathcal{M}_f^*)$ from the input ADEs as in \eqref{eq:radratdyn}.
    \item Denote by $E$ the set of polynomials 
        \begin{eqnarray*}
             E&\coloneqq &\lbrace Q\,{\mathbf{w}'}^{\mu}-\mathbf{a}(\mathbf{w})-\mathbf{e}_{\mathbf{a}}(\mathbf{w}'), Q \, z - b(\mathbf{w})\rbrace\\
               &=& \lbrace Q\,{w_i'}^{\mu_i}-a_i(w_1,\ldots,w_M) - e_{a_i}(w_i'),i=1,\ldots,M,\, Q\,z - b(w_1,\ldots,w_M)\rbrace 
        \end{eqnarray*}
        in the polynomial ring $\KK[x,w_1,\ldots,w_M,w_1',\ldots,w_M',z]$.
        \item Compute the first $M-1$ derivatives (w.r.t. $d/dx$) of all polynomials in $E$ and add them to $E$.
        \item Compute the $M$th derivative of $Q\, z - b(w_1,\ldots,w_M)$ and add it to $E$. We are now in the ring $\KK[x,w_1^{(\leq M)},\ldots,w_M^{(\leq M)},z^{(\leq M)}]$.
        \item Let  $I\coloneqq\langle E \rangle \subset \KK[x,w_1^{(\leq M)},\ldots,w_M^{(\leq M)},z^{(\leq M)}]$ be the ideal generated by the elements of~$E$.
        \item If some of the input are not l.h.o., let $H$ be the least common multiple of $Q$ and the separants of $Q{\mathbf{w}'}^{\mu}-\mathbf{a}(\mathbf{w}) - \mathbf{e}_{\mathbf{a}}(\mathbf{w}')$. Otherwise $H\coloneqq Q$.
        \item Update $I$ by its saturation with $H$, i.e, $I\coloneqq I\colon H^{\infty}$.
        \item\label{step:elimalg1} Compute the elimination ideal $I \cap \KK[x][z^{(\leq M)}]$. From the Gr\"{o}bner basis, choose a polynomial $q$ of the lowest degree among those of the lowest order. 
        \item Return $q=0$ (or its writing with $d/dx$).
    \end{enumerate}	
    \end{algorithmic}
\end{algorithm}

To prove the correctness of \Cref{algo:Algo1}, we must show that the elimination ideal $I_{\mathcal{M}_f^*}^{\leq M} \cap \KK[x][z^{(\infty)}]$ in \cref{step:elimalg1} is non-trivial. This fact is established by the following theorem.

\begin{theorem}\label{lem:lem1} On the commutative ring $\KK(x)[\mathbf{w}^{(\infty)},z^{(\infty)}]$, seen as a polynomial ring in infinitely many variables, consider the lexicographic monomial ordering corresponding to any ordering on the variables such that
\begin{itemize}
\item[(i.)] ${z}^{(j_1)}\succ {w_i}^{(j_2)}$ for all $i, j_1, j_2\in \NN$,
\item[(ii.)] $z^{(j+1)}\succ z^{(j)}$ and $w_{i_1}^{(j+1)}\succ w_{i_2}^{(j)}$ for all $i_1, i_2,j\in\NN$.
\end{itemize}
Then the set $E\coloneqq\left\lbrace Q{\mathbf{w}'}^{\mu}-\mathbf{a}(\mathbf{w})-\mathbf{e}_{\mathbf{a}}(\mathbf{w}'),\, Qz-b(\mathbf{y}) \right\rbrace$ is a triangular set with respect to this ordering. Moreover,
\begin{equation}\label{eq:eq10}
    I_{\mathcal{M}_f^*}^{\leq M} \cap \KK[x][z^{(\infty)}] \neq \langle 0 \rangle.
\end{equation}
\end{theorem}
\begin{proof} The leading monomials of $\left(Q\,{w_i'}^{\mu_i}-a_i(\mathbf{w})-e_{a_i}(w_i')\right)^{(j)}$ and $\left(Q\,z-b(\mathbf{w})\right)^{(j)}$ in the ring $\KK(x)[\mathbf{w}^{(\infty)},z^{(\infty)}]$ (with derivation seen in $S_{\mathbf{w},z}$) have highest variables $w_i^{(j+1)}$ and $z^{(j)}$, respectively. Since these variables are all distinct, by definition (see \cite[Definition 4.1]{hubert2003notes}), we deduce that $E$ is a consistent triangular set. We shall see the coefficients in the field $\KK(x)$. As a triangular set, $E$ defines the ideal $\langle E\rangle:H^{\infty}=I_{\mathcal{M}_f^*}$. Therefore by \cite[Theorem 4.4]{hubert2003notes}) all associated primes of $I_{\mathcal{M}^*}^{\leq M}$ share the same transcendence basis given by the non-leading variables $\{w_1,\ldots,w_M\}$ in $I_{\mathcal{M}_f^*}^{\leq M}$. Thus the transcendence degree of $\KK(x)[\mathbf{w}^{(\leq M)},z^{(\leq M)}]/I_{\mathcal{M}_f^*}^{\leq M}$ over $\KK(x)$ is $M$. However, the transcendence degree of $\KK(x)[z^{(\leq M)}]$ is $M+1$. Hence we must have $I_{\mathcal{M}_f^*}^{\leq M} \cap \KK[x][z^{(\infty)}] \neq \langle 0 \rangle$. \CQFD
\end{proof}
From the proof of \Cref{lem:lem1}, one should observe that $M$ is the minimal integer for which the used arguments hold. This relates to the minimality of the order of the ADE obtained after the computations. Moreover, the l.h.o and the non-l.h.o cases are unified in \eqref{eq:radratdyn}. We make $\mathbf{\mathcal{E}}$ explicit to illustrate the difference with the construction of $(\mathcal{M}_f)$ in \Cref{sec:lho}.

 The main point behind \Cref{algo:Algo1} is the systematic construction of ADEs for rational expressions of D-algebraic functions using ODE systems (or triangular sets). Note that the saturation with respect to the least common multiple of the denominators and the separants implies that the D-algebraic functions we deal with are the generic solutions of ADEs.

 \subsection{Implementation}

We have implemented \Cref{algo:Algo1} in the \texttt{NLDE} package (see \cite{BTNLDE,NLDE}). It complements the implementation for arithmetic operations from \cite{RSB2023}. For full details on the syntax, we refer the reader to the documentation at \href{https://mathrepo.mis.mpg.de/OperationsForDAlgebraicFunctions/NLDEdoc.html}{MathRepo NLDE Doc} (see also \cite{NLDE}). We here present examples to illustrate what can be seen as an improvement of the implementation in \cite{RSB2023}. The commands for arithmetic operations are \texttt{NLDE:-arithmeticDalg} and \texttt{NLDE:-unaryDalg}. The former is for several input ADEs, and the latter is for a single input ADE. By default, \texttt{NLDE:-arithmeticDalg} follows the second method in \cite{RSB2023}, which implies that in the non-l.h.o. case, the order of the returned ADE may be higher than expected. When some of the input ADEs are not l.h.o., one specifies \texttt{lho=false} to use \Cref{algo:Algo1} and obtain an ADE of the smallest order possible. Regarding Gr\"obner bases elimination, by default, 
\begin{itemize}
    \item when \texttt{lho=true}, we use the \textit{pure lexicographic} ordering \texttt{plex} from the Maple package \texttt{Groebner}; this is the same ordering used in \Cref{lem:lem1}, which guarantees the desired result.
    \item When \texttt{lho=false}, we use the \textit{lexdeg} elimination order of Bayer and Stillman \cite{bayer1987theorem} implemented by the \texttt{EliminationIdeal} command of the  Maple package \texttt{PolynomialIdeals}. This is a block order where blocks are ordered with \texttt{plex}. In our case, there are two blocks: the variables of the elimination ideal constitute one block, and the remaining variables constitute the other one. Both blocks are internally ordered with a suitable order like \texttt{degrevlex} for instance. Although this ordering generally returns the desired result, it can also fail to find it sometimes. Nevertheless, \textit{lexdeg} tends to provide a better efficiency. To use \texttt{plex} when \texttt{lho=false}, one further specifies \texttt{lhoplex=true}.
\end{itemize}
To use the \textit{lexdeg} ordering when \texttt{lho=true}, one specifies \texttt{ordering=lexdeg}. For unary operations, \texttt{NLDE:-unaryDalg} was updated to only implement the method of this paper.
\begin{example}\label{eg:eg1} Consider the ADE 
\begin{equation}\label{eq:ade1}
    {y_1'(x)}^2 + {y_1(x)}^2=1.
\end{equation}
Its solutions are $-1, 1,$ and $\lambda_1\,\cos(x) + \lambda_2 \sin(x)$ for arbitrary $\lambda_1,\lambda_2\in\KK$ such that $\lambda_1^2+\lambda_2^2=1$. As a second ADE, we take the ODE of the exponential function
\begin{equation}\label{eq:ade2}
y_2'(x)=y_2(x).
\end{equation}
We want to compute ADEs for the sum of solutions of these two ADEs. Using Method II from \cite{RSB2023}, one gets the output:
\begin{lstlisting}
> ADE1:=diff(y[1](x),x)^2+y[1](x)^2-1=0:
> ADE2:=diff(y[2](x),x)=y[2](x):
> Out1:=NLDE:-arithmeticDalg([ADE1,ADE2],[y[1](x),y[2](x)],z=y[1]+y[2])
\end{lstlisting}
\begin{small}
\begin{dmath}\label{eq:out1}
\mathit{Out1} \coloneqq \frac{d^{3}}{d x^{3}}z \! \left(x \right)-\frac{d^{2}}{d x^{2}}z \! \left(x \right)+\frac{d}{d x}z \! \left(x \right)-z \! \left(x \right)=0.
\end{dmath}
\end{small}
Since \eqref{eq:ade1} is not l.h.o., its first derivative is used; which explains why the output ADE is of order $2+1=3$. Let us now use the method of this paper. The corresponding ODE system is
\begin{equation}\label{eq:odesys1}
\begin{cases}
     {w_1'}^2 = - {w_1}^2 - 1\\
     w_2'     = w_2\\
     z        = w_1 + w_2
\end{cases}.
\end{equation}
Thus $\mu=(2,1)^T, \mathbf{\mathcal{E}}=(0,0)^T$.
We use our package as follows:
\begin{lstlisting}
> Out2:=NLDE:-arithmeticDalg([ADE1,ADE2],[y[1](x),y[2](x)],z=y[1]+y[2],lho=false):
\end{lstlisting}
\vspace{-0.5cm}

\begin{dmath}\label{eq:out2}
\left(\frac{d^{2}}{d x^{2}}z \! \left(x \right)\right)^{2}-2 \left(\frac{d}{d x}z \! \left(x \right)\right) \left(\frac{d^{2}}{d x^{2}}z \! \left(x \right)\right)
+2 \left(\frac{d}{d x}z \! \left(x \right)\right)^{2}-2 z \! \left(x \right) \left(\frac{d}{d x}z \! \left(x \right)\right)+z \! \left(x \right)^{2}-2=0.
\end{dmath}
We obtain a second-order ADE satisfied by sums of all generic solutions to \eqref{eq:ade1} and \eqref{eq:ade2}. Notice that $1+\exp(x)$ is not a solution of \eqref{eq:out1} and \eqref{eq:out2}. As explained in \cite{RSB2023}, in general, the saturation often neglects polynomial solutions of degrees less than the order of non-l.h.o. ADEs; this applies to \eqref{eq:out1}. Similarly, for \eqref{eq:out2}, saturation at the separants eliminates those same polynomials.

On the other hand, our implementation provides an option for the user to avoid saturation by the separants. The optional argument is \texttt{separantsZeros=true}, which is \texttt{false} by default. The option is exploratory and does not rely on \Cref{lem:lem1} when \texttt{separantsZeros=true}. In this case, we have no guarantee that the corresponding elimination ideal is non-trivial. Nevertheless, we have not yet found an example where this does not happen.

\begin{lstlisting}
> Out2:=NLDE:-arithmeticDalg([ADE1,ADE2],[y[1](x),y[2](x)],z=y[1]+y[2],
    lho=false,separantsZeros=true):
> factor(Out2)
\end{lstlisting}
\vspace{-0.5cm}

\begin{dmath}\label{eq:out2_2}
\left(\frac{d}{d x}z \! \left(x \right)-z \! \left(x \right)+1\right) \left(\frac{d}{d x}z \! \left(x \right)-z \! \left(x \right)-1\right) \left(\left(\frac{d^{2}}{d x^{2}}z \! \left(x \right)\right)^{2}-2 \left(\frac{d}{d x}z \! \left(x \right)\right) \left(\frac{d^{2}}{d x^{2}}z \! \left(x \right)\right)\\
+2 \left(\frac{d}{d x}z \! \left(x \right)\right)^{2}-2 z \! \left(x \right) \left(\frac{d}{d x}z \! \left(x \right)\right)+z \! \left(x \right)^{2}-2\right)=0.
\end{dmath}
The particular solutions $1+\lambda \exp(x)$ and $-1+\lambda \exp(x),$ $\forall\,\lambda \in\KK$, are zeros of the factors $\left(\frac{d}{d x}z \! \left(x \right)-z \! \left(x \right)+1\right)$ and $\left(\frac{d}{d x}z \! \left(x \right)-z \! \left(x \right)-1\right)$, respectively.\QEDA 
\end{example}

\begin{example}\label{eg:eg2} 
Let us take the third ADE:
\begin{equation}\label{eq:ade3}
    {y_3'(x)}^3+{y_3'(x)}^2+3=0.
\end{equation}
Its solutions are $\alpha_i\,x + \lambda$, $i=1,2,3$, where $\lambda$ is an arbitrary constant, and the $\alpha_i$'s are roots of the polynomial $X^3+X^2+3=0$. We want to find an ADE for $\frac{f_1\,f_3}{f_2}$, where $f_1,f_2$ and $f_3$ are solutions of \eqref{eq:ade1}, \eqref{eq:ade2}, and \eqref{eq:ade3}, respectively.

The corresponding ODE system is
\begin{equation}\label{eq:odesys2}
    \begin{cases}
     {w_1'}^2 = - {w_1}^2 - 1\\
     w_2'     = w_2\\
     {w_3'}^3 = -3-{w_3'}^2\\
     z        = \frac{w_1\,w_3}{w_2}
    \end{cases}.
\end{equation}
Thus $\mu=(2,1,3)^T, \mathbf{\mathcal{E}}=(0,0,-{w_3'}^2)^T$. The implementation of \Cref{algo:Algo1} yields:
\begin{lstlisting}
> Out4:=NLDE:-arithmeticDalg([ADE1,ADE2,ADE3],
   [y[1](x),y[2](x),y[3](x)],z=y[1]*y[3]/y[2],lho=false):
\end{lstlisting}
\vspace{-0.4cm}

\begin{dmath}\label{eq:out3}
\mathit{Out3} \coloneqq 12 z \! \left(x \right)^{2}+32 \left(\frac{d}{d x}z \! \left(x \right)\right) z \! \left(x \right)+20 \left(\frac{d^{2}}{d x^{2}}z \! \left(x \right)\right) z \! \left(x \right)+6 \left(\frac{d^{3}}{d x^{3}}z \! \left(x \right)\right) z \! \left(x \right)+24 \left(\frac{d}{d x}z \! \left(x \right)\right)^{2}+30 \left(\frac{d^{2}}{d x^{2}}z \! \left(x \right)\right) \left(\frac{d}{d x}z \! \left(x \right)\right)+10 \left(\frac{d^{3}}{d x^{3}}z \! \left(x \right)\right) \left(\frac{d}{d x}z \! \left(x \right)\right)+9 \left(\frac{d^{2}}{d x^{2}}z \! \left(x \right)\right)^{2}+6 \left(\frac{d^{3}}{d x^{3}}z \! \left(x \right)\right) \left(\frac{d^{2}}{d x^{2}}z \! \left(x \right)\right)+\left(\frac{d^{3}}{d x^{3}}z \! \left(x \right)\right)^{2}=0.
\end{dmath}

For this example, \cite[Method II]{RSB2023} gives the same output. The latter is one of the three factors of the ADE obtained with the option \texttt{separantsZeros=true}. The other two factors are the following:

\begin{dmath}\label{eq:out4}
\left(\frac{d^{2}}{d x^{2}}z \! \left(x \right)+2 \frac{d}{d x}z \! \left(x \right)+z \! \left(x \right)\right) \left(\frac{d^{2}}{d x^{2}}z \! \left(x \right)+2 \frac{d}{d x}z \! \left(x \right)+2 z \! \left(x \right)\right)^{2}.
\end{dmath}
These two factors vanish at $\lambda\,\alpha_i\,\exp(-x)$, $i=1,2,3$, $\forall\, \lambda\in\KK$.
\QEDA
\end{example}

 Note that either of the implementations of the Rosendfeld-Gr\"{o}bner algorithm and the Thomas decomposition could not find \eqref{eq:out3} within a reasonable time. The former runs out of memory, and the latter keeps running even after an hour of computations. They tend to behave this way for ADEs of higher degrees. We also mention that even in the situation where one of these algorithms would return a correct result, the exponent in \eqref{eq:out4} could not be found since those algorithms output generators of a radical differential ideal, unlike our algorithm which computes in the corresponding differential ideal. It would be interesting to know what these exponents mean.

\section{Arithmetic of multivariate D-algebraic functions}\label{sec:multivar}

We now focus on solutions of algebraic PDEs. We aim to present an algorithm in the spirit of Method I in \cite{RSB2023}. For simplicity, we start with bivariate D-algebraic functions. Thus the $N$ functions $f_1,\ldots,f_N$ are now bivariate in $\mathbf{x}=x_1,x_2$, and each $f_i$ is a zero of the partial differential polynomial $p_i\in R[y_i^{(\infty,\infty)}]$ of order $n_{i,1}$ w.r.t. $x_1$, and of order $n_{i,2}$ w.r.t. $x_2$. 

\subsection{Ordering on the semigroup of derivative operators}\label{sec:derrule}

We are going to introduce a total ordering on the set of derivative operators with the aim of using them in increasing order up to a given limit. Algorithmically, this will entail to defining a map that sends every non-negative integer to a derivative operator. In \cite{kolchin1973differential} and many references in differential algebra (see also \cite{freitag2020effective}), the notion of order for partial differential rings is different from the one we use. The prevalent definition does not easily adapt to our computational goal. Indeed, given the set of derivations $\Delta_2 = \{\partial_{x_1}, \partial_{x_2}\}$, one considers the free multiplicative semigroup $\Theta_2$ generated by the elements of $\Delta_2$. To any derivative operator 
$$\theta\coloneqq \partial_{x_1}^{n_1}\,\partial_{x_2}^{n_2}\in\Theta,~ n_1,n_2\in\NN,$$
one defines the order $s=n_1 + n_2$. The derivative $\theta u$ of $u\in R$, is called the derivative of $u$ of order $s$ (supposedly w.r.t. $\theta$). The order of a differential polynomial $p\in R[y^{(\infty,\infty)}]$ (also called $\Delta_2$-polynomial ring) is then the maximum $n_i + n_j$ such that $y^{(n_i,n_j)}$ occurs in $p$. What we wish to add in this definition is a total ordering (or ranking) of the derivative operators in $\Theta_2$. Observe that for $\Delta_1=\{\partial_x\}$, there is a natural ranking of the derivative operators:
$$
\{\partial_x^0, \partial_x^1, \partial_x^2,\ldots\},
$$
allowing to \textit{uniquely} associate a non-negative integer with the order of a particular derivative. With the above definition of order for bivariate partial differential polynomial rings, there are two possible choices for the derivative of $y$ of order $1$, namely, $y^{(0,1)}$ and $y^{(1,0)}$. We combine total order and classical orderings (lexicographic or reverse-lexicographic).

The total-order is used to ensure the possibility to move from $y^{(n_1,0)}$ to $y^{(0,n_2)}$ by some applications of the chosen derivation rule. In fact, the \textit{bijective map} between $\NN$ and $\NN^2$ should be ``obvious'' by construction. Thus what we call order for $y^{(n_1,n_2)}$ is either the non-negative integer whose image (or preimage) is $(n_1,n_2)$ through that bijection, or simply $(n_1,n_2)$, assuming $\Theta_2$ is well-ordered (might be implicit for non-computational purposes). We will use the latter, and define a derivation rule based on a bijective map between $\NN^2$ and $\NN$. 

Our idea is to use a ``very'' classical map in set theory. Consider the following start of a ranking for the application of derivative operators 

$$
\partial_{x_1} \longrightarrow \partial_{x_2} \longrightarrow \partial_{x_1}^2 \longrightarrow \partial_{x_1} \partial_{x_2}\longrightarrow \partial_{x_2}^2 \longrightarrow \partial_{x_1}^3 \ldots .
$$

This corresponds to the following ranking of derivative orders

$$
(1,0) \longrightarrow (0,1) \longrightarrow (2,0) \longrightarrow (1,1) \longrightarrow (0,2) \longrightarrow (3,0) \ldots .
$$

One can associate each ordered pair of integers with a unique positive integer. Indeed, the underlying map is the so-called Cantor pairing function, which is a bijection from $\NN^2$ to $\NN$. The underlying ordering of the ordered pairs is the graded (or degree) lexicographic ordering; however, to avoid confusion with the lexicographic ordering that we will use for monomials, we use the terminology of Cantor pairing function or Cantor $l$-tuple functions for arbitrarily many independent variables. As proven by Rudolf Fueter and George Pólya \cite{fueterPolya1923rationale}, the Cantor pairing function is the only quadratic polynomial that maps $\NN^2$ to $\NN$. Therefore for two variables, it appears as the most ``natural and efficient'' choice for processing speed on a regular computer.

We see that such a ranking is inherent to the definition of a total order on $\Theta_2 $ to rank the variable in $S_y$. Note that this ranking is in line with the terminology introduced by Ritt and should not be confused with monomial orderings. Let $\sigma_2$ be the Cantor pairing function, $\sigma_2^{-1}$ its functional inverse such that $\sigma_2^{-1}(k) = (\sigma_{2,1}^{-1}(k),\sigma_{2,2}^{-1}(k))$. For $k\in\NN$, the $k$th derivation, denoted $\theta_{x_1,x_2}^k$, is:
\begin{equation}\label{eq:thetak}
    \begin{split}
        &\theta_{x_1,x_2}^k \colon S_{y} \longrightarrow S_{y}\\
        &y^{(i,j)} \mapsto \partial_{x_1}^{\sigma_{2,1}^{-1}(k)}\partial_{x_2}^{\sigma_{2,2}^{-1}(k)} y^{(i,j)} = y^{(i+\sigma_{2,1}^{-1}(k), j+\sigma_{2,2}^{-1}(k))},\, i,j\in\NN.
    \end{split}
\end{equation}
We view $R[y^{(\infty,\infty)}]$ as 
$$
R[\theta_{x_1,x_2}^{\infty} y^{(0,0)}] \coloneqq R[\theta_{x_1,x_2}^1 y^{(0,0)},\theta_{x_1,x_2}^2y^{(0,0)},\ldots] = R[y^{(0,0)},y^{(1,0)},y^{(0,1)},y^{(2,0)},\ldots].
$$
This implicitly defines an operator which we denote by $\theta_{x_1,x_2}$.

\begin{remark} \item
\begin{itemize}
    \item For all $k\in \NN$,  $\theta_{x_1,x_2}^k \in \Theta_2$, and therefore satisfies the linearity and the Leibniz rules.
    \item On the operator side, $\theta_{x_1,x_2}^k$ may wrongly be seen as the $k$th power of the operator $\theta_{x_1,x_2}$. Although commutativity is induced by the semigroup $\Theta_2$, exponentiation does not follow. In general, for $i,j\in\NN$,
    \begin{equation}\label{eq:nonassoc}
         \theta_{x_1,x_2}^{i+j}\neq \theta_{x_1,x_2}^i \theta_{x_1,x_2}^j = \theta_{x_1,x_2}^j \theta_{x_1,x_2}^i.
    \end{equation}
    The equality is easy to establish since for $\theta_{x_1,x_2}^i=\partial_{x_1}^{i_1}\partial_{x_2}^{i_2}$, and $\theta_{x_1,x_2}^j=\partial_{x_1}^{j_1}\partial_{x_2}^{j_2}$, we have
    $$ \theta_{x_1,x_2}^i \theta_{x_1,x_2}^j = \partial_{x_1}^{i_1+j_1}\partial_{x_2}^{i_2+j_2} = \theta_{x_1,x_2}^j \theta_{x_1,x_2}^i. $$
    Just to give an example, let $p:=y^{(1,2)}$. Then we have
    \begin{equation}
        \theta_{x_1,x_2}^8(p)=y^{(2,4)} \neq \theta_{x_1,x_2}^5\left(\theta_{x_1,x_2}^3(p)\right)= \theta_{x_1,x_2}^3\left(\theta_{x_1,x_2}^5(p)\right)=y^{(3,4)}=\theta_{x_1,x_2}^{12}(p)
    \end{equation}
    
    We will consider the $k$th derivation w.r.t. $\theta_{x_1,x_2}$ as the application of $\theta_{x_1,x_2}^k$, and not $k$ successive applications of $\theta_{x_1,x_2}^1=\partial_{x_1}$.
    %\item The ordering on the set of derivative operators may be used to define a derivation in $S_y$. Indeed, by assuming that ${y^{(i,j)}}^{'}$ is $\theta_{x_1,x_2}^{k+1}y^{(0,0)}$, where $k$ is such that $\theta_{x_1,x_2}^ky^{(0,0)}=y^{(i,j)}$. The derivation is then defined as the linear map that fulfills the Leibniz rule and replaces every variable by its successor according to that ordering, and every constant by 0. We will not have a particular interest on this derivation for the remaining par of the paper, but we  
    %\item The linearity and the Leibniz rules are fulfilled by every $\theta_{x_1,x_2}^k, k\in\NN$.
    %There is another view of $\theta_{x_1,x_2}$ that bestows it the property of being a derivation as well as $\partial_{x_1}$ and $\partial_{x_2}$ on $S_y$.
    %\item $\theta_{x_1,x_2}$ may be seen as a derivative operator for the differential polynomial ring $R[y^{(\infty,\infty)}]$. However, it must take two arguments: the order, and the ring element. The order is necessary to track derivations. This is mainly because we do not have commutativity on the operator side; 
\end{itemize}
\end{remark}

We introduced $\theta_{x_1, x_2}$ primarily to avoid continuous referencing to the chosen ordering. This will significantly simplify the arguments in the upcoming sections. Needless to say, one can easily avoid explicit links to this operator.

\subsection{Arithmetic of bivariate D-algebraic functions}

 Observe that, unlike other definitions of multivariate D-algebraic functions that restrict their study to D-algebraic functions in each independent variable (see \cite[Section 5.2]{vdH19},\cite[Section 6]{raschel2020counting}), we here have a more general notion. Indeed, a multivariate function can be D-algebraic without being D-algebraic in some of its variables. One can derive proofs for the closure properties by providing bounds for the transcendence degree, as shown in \cite[Section 5.2]{vdH19}. Note, however, that the bounds may not be as sharp as in the ordinary case.
\begin{example} The incomplete $\Gamma$ function $\Gamma_I(x_1,x_2)$ defined as
\begin{equation}\label{eq:incGammadef}
\Gamma_I(x_1,x_2) = \Gamma(x_1) - \frac{x_2^{x_1}\, _1F_1\left(\begin{matrix} x_1 & x_1+1 \end{matrix}, - x_2\right)}{x_1},
\end{equation}
where $_1F_1$ denotes the confluent hypergeometric function, is D-algebraic, and satisfies the linear ADE:
\begin{dmath}\label{eq:incGammaeq}
\left(-x_{1}+x_{2}+1\right) \frac{\partial}{\partial x_{2}}y \! \left(x_{1},x_{2}\right)+x_{2} \frac{\partial^{2}}{\partial x_{2}^{2}}y \! \left(x_{1},x_{2}\right)=0.
\end{dmath}\QEDA
\end{example}
Getting back to our algorithmic goal, we can now write every differential polynomial in $R[y^{(\infty,\infty)}]$ in terms of $\theta_{x_1,x_2}$ derivations.
\begin{example} Since every $y^{(n_1,n_2)}$ is the result of $\theta_{x_1,x_2}^k y^{(0,0)}$, for $k=\sigma_2(n_1,n_2)\in\NN$, we can write $\theta_2^ky$, replacing $\{x_1,x_2\}$ by $2$ and omitting the $(0,0)$ superscript. One can even remove $y$ and the subscript $2$ if there is no ambiguity for the variables involved (and get `polynomials' in $\theta$).
\begin{align}
   &x_1\,(y^{(1,2)})^2\,y^{(0,1)}+x_2\,(y^{(0,0)})^3\,y^{(3,0)}- (y^{(3,1)})^2 = x_1\, (\theta_2^8y)^3\, \theta_2^2y + x_2\,(\theta_2^0y)^3\,\theta_2^6y - (\theta_2^{11}y)^2\\
   &x_1^3\,(y^{(1,1)})^4\,y^{(4,1)} + x_1^3\,x_2^2\,y^{(1,3)} = (\theta_2^4y)^4\,\theta_2^{16}y + x_1^3\,x_2^2\,\theta_2^{13}y
\end{align}
\QEDA
\end{example}
% ctrl + /
 \begin{remark}\item
\begin{itemize}
    \item $\theta_2^1 x_1 = 1,\, \theta_2^k x_1 = 0,\, \forall k \geq 2$.
    \item $\theta_2^2 x_2 = 1,\, \theta_2^k x_2 = 0, \,\forall k \neq 2$.
\end{itemize}
\end{remark}
\begin{proposition}[Composition property]\label{prop:prop2} For all non-negative integers $k,l$, and a differential indeterminate $y$, we have
\begin{equation}
    \theta_2^k \left( \theta_2^l\,y\right) = \theta_2^{\sigma_2\left(\sigma_2^{-1}(k)+\sigma_2^{-1}(l)\right)} y,
\end{equation}    
where addition for ordered pairs is taken component wise.
\end{proposition}
\begin{proof}
    Let $y^{(i,j)}=y^{\sigma_2^{-1}(l)}$ be the result of $\theta_2^ly$. Then according to \eqref{eq:thetak}, 
    $$\theta_2^k y^{(i,j)}=y^{\left(i+\sigma_{2,1}^{-1}(k),j+\sigma_{2,2}^{-1}(k)\right)} = y^{\sigma_2^{-1}(l)+\sigma_2^{-1}(k)}=\theta_2^{\sigma_2\left(\sigma_2^{-1}(k)+\sigma_2^{-1}(l)\right)} y. $$  \CQFD
\end{proof}
To illustrate our use of $\theta_2$, let us take an explanatory example.

\begin{example}\label{ex:illexpl1} Consider
$$
p_1=y_1^{(0,1)}+x_2\,y_1^{(1,1)},\,\text{ and }\, p_2=x_1\,y_2^{(1,0)}-y_2^{(2,0)}.
$$
With $\theta_2$, we have
\begin{align}
&p_1 \coloneqq \theta_2^2y_1+x_2\,\theta_2^4y_1,\\
&p_2 \coloneqq x_1\,\theta_2^1y_2 - \theta_2^3y_2.
\end{align}
Thus w.r.t. $\theta_2$, $p_1$ and $p_2$ are of order $4$ and $3$, respectively. Suppose we want to find an algebraic PDE for the sum $f_1 + f_2$, where $f_i$ is a zero of $p_i,i=1,2$. Our knowledge of the ordinary case suggests that we look for an ADE of order at most $3$ in $x_1$ and $1$ in $x_2$. This is the same as looking for an ADE of $\theta_2$-order at most $11$. We define $z(x_1,x_2)$ (representing $(y_1(x_1,x_2)+y_2(x_1,x_2)$) as the dependent variable for the targeted ADE, and consider the differential polynomial
$$
q\coloneqq z-y_1-y_2\in S_{y_1,y_2,z}.
$$
We take the $11$ first $\theta_2$-derivatives of $q$ and remove those that involve derivatives of orders greater than $3$ w.r.t. $x_1$ (first component) and greater than $1$ w.r.t. $x_2$ (second component). We get
\begin{dmath}\label{eq:derq}
z^{\left(1,0\right)}-y_{1}^{\left(1,0\right)}-y_{2}^{\left(1,0\right)},z^{\left(0,1\right)}-y_{1}^{\left(0,1\right)}-y_{2}^{\left(0,1\right)},z^{\left(2,0\right)}-y_{1}^{\left(2,0\right)}-y_{2}^{\left(2,0\right)},z^{\left(1,1\right)}-y_{1}^{\left(1,1\right)}-y_{2}^{\left(1,1\right)},z^{\left(3,0\right)}-y_{1}^{\left(3,0\right)}-y_{2}^{\left(3,0\right)},z^{\left(2,1\right)}-y_{1}^{\left(2,1\right)}-y_{2}^{\left(2,1\right)},z^{\left(3,1\right)}-y_{1}^{\left(3,1\right)}-y_{2}^{\left(3,1\right)}.
\end{dmath}
We now need to take some $\theta_2$-derivatives of $p_1$ and $p_2$. The idea is to make the highest $\theta_2$-derivatives of $y_1$ and $y_2$ from \eqref{eq:derq} occur. Here we wish to have $\theta_2^{11} y_i = y_i^{(3,1)},i=1,2,$ among the derivatives. Our bound for the number of derivations to be taken is $11$. For this example, it turns out that we only need to take the first $4$ $\theta_2$-derivatives. For $p_1$, we obtain the differential polynomials
\begin{dmath}\label{eq:derp1}
x_{2} y_{1}^{\left(2,1\right)}+y_{1}^{\left(1,1\right)},x_{2} y_{1}^{\left(1,2\right)}+y_{1}^{\left(0,2\right)}+y_{1}^{\left(1,1\right)},x_{2} y_{1}^{\left(3,1\right)}+y_{1}^{\left(2,1\right)},x_{2} y_{1}^{\left(2,2\right)}+y_{1}^{\left(1,2\right)}+y_{1}^{\left(2,1\right)}.
\end{dmath}
And for $p_2$ we obtain
\begin{dmath}\label{eq:derp2}
    x_{1} y_{2}^{\left(2,0\right)}+y_{2}^{\left(1,0\right)}-y_{2}^{\left(3,0\right)},x_{1} y_{2}^{\left(1,1\right)}-y_{2}^{\left(2,1\right)},x_{1} y_{2}^{\left(3,0\right)}+2 y_{2}^{\left(2,0\right)}-y_{2}^{\left(4,0\right)},x_{1} y_{2}^{\left(2,1\right)}+y_{2}^{\left(1,1\right)}-y_{2}^{\left(3,1\right)}.
\end{dmath}
We write these differential polynomials with their dependent variables and not in terms of $\theta_2$-derivations. This is a preliminary step for the Gr\"obner basis elimination that will follow. Note that these computations can be performed without any explicit use of $\partial_{x_1}$ and $\partial_{x_2}$, or any derivative operator in $\Theta_2$. Indeed, we have the natural isomorphism 
\begin{equation}
    \left(\KK(x_1,x_2),\Delta_2\right) \cong \left(\KK(x_1,x_2), \theta_{x_1,x_2}\right), \label{pagerank}
\end{equation}
by the correspondence between $\theta_2$-derivations and elements of $\Theta_2$.

In a similar way as defined in \Cref{lem:lem1}, we consider a lexicographic monomial ordering such that
\begin{itemize}
\item[(I.)] $\theta_2^{k_2} z \succ \theta_2^{k_1} y_i$ for all $i\in\{1,2\}, k_1, k_2\in \NN$,
\item[(II.)] $\theta_2^{k+1} z \succ \theta_2^k z$ and $\theta_2^{k+1} y_{i_1}\succ\theta_2^k y_{i_2}$ for all $i_1,i_2\in\{1,2\}, k\in\NN$.
\end{itemize}
We want to eliminate the lower variables $\theta_2^k y_1, \theta_2^k y_2,k\in\NN$. We compute the elimination ideal
\begin{equation}
 I_{p_1,p_2,q}^{4,4,(3,1)} \cap \KK[x_1,x_2][z^{\left(0,0\right)},z^{\left(1,0\right)},z^{\left(0,1\right)},,z^{\left(2,0\right)},z^{\left(1,1\right)},z^{\left(3,0\right)},z^{\left(2,1\right)},z^{\left(3,1\right)}],
\end{equation}
where $I_{p_1,p_2,q}^{4,4,(3,1)}$ is the ideal containing $p_1$, $p_2$, $q$, and their derivatives from \eqref{eq:derp1}, \eqref{eq:derp2}, and \eqref{eq:derq}. This may be viewed as a ramification of a truncation of the corresponding differential ideal in $S_{y_1,y_2,z}$ with the ranking imposed by $\theta_2$; however, the truncation would be of $\theta_2$-order $11$. We obtain the principal ideal
\begin{equation}\label{eq:elimIpq}
    \left\langle \left(x_{1}^{2} x_{2}+x_{1}+x_{2}\right) z^{\left(1,1\right)}+\left(x_{1}^{2} x_{2}^{2}+x_{2}^{2}-1\right) z^{\left(2,1\right)}+\left(-x_{1} x_{2}^{2}-x_{2}\right) z^{\left(3,1\right)}\right\rangle.
\end{equation}
Therefore the sum of zeros of $p_1$ and $p_2$ are solutions of the associated linear algebraic PDE given by
\begin{dmath}\label{eq:addeq1}
\left(x_{1}^{2} x_{2}+x_{1}+x_{2}\right) \left(\frac{\partial^{2}}{\partial x_{1}\partial x_{2}}z \! \left(x_{1},x_{2}\right)\right)+\left(x_{1}^{2} x_{2}^{2}+x_{2}^{2}-1\right) \left(\frac{\partial^{3}}{\partial x_{1}^{2}\partial x_{2}}z \! \left(x_{1},x_{2}\right)\right)+\left(-x_{1} x_{2}^{2}-x_{2}\right) \left(\frac{\partial^{4}}{\partial x_{1}^{3}\partial x_{2}}z \! \left(x_{1},x_{2}\right)\right)=0.
\end{dmath} \QEDA
\end{example}

In \Cref{ex:illexpl1}, we obtained an ADE whose order is the sum of the orders of $p_1$ and $p_2$ w.r.t. each independent variable. A natural question that we might want to answer is to know whether, like in the ordinary case, the sum of orders of given differential polynomials constitutes a bound for the order of ADEs fulfilled by functions defined as arithmetic expressions of zeros of those polynomials. We will show that this does not hold in general. We give another example to that aim, a tiny modification of \Cref{ex:illexpl1}.

\begin{example}[The orders do not add up]

Let
\begin{equation}\label{eq:eg2p1p2}
p_1\coloneqq x_1\,y_1^{(0,1)}+x_2\,y_1^{(1,1)},\,\text{ and }\, p_2\coloneqq x_1^2\,y_2^{(1,0)}-y_2^{(2,0)}\,x_2.
\end{equation}
As previously, we can expect to find a partial differential polynomial of order at most $3$ in $x_1$ and $1$ in $x_2$ for the sum of zeros of $p_1$ and $p_2$. However, for this example, the elimination ideal is trivial, even when we do the elimination step with the truncation of $\theta_2$-order $11$ of the corresponding differential ideal.

We now look for an algebraic PDE of order $4$ in the first component and order $1$ in the second. This corresponds to the $\theta_2$-order $23$. For the polynomials $p_1$ and $p_2$, we take their first $7$ $\theta_2$-derivatives and proceed as before. Then we compute the elimination ideal
\begin{equation}\label{eq:elimeg2}
      \left\langle \left\{\theta_2^{\leq 7} p_1,\, \theta_2^{\leq 7} p_2, \theta_2^{\leq 23} q\right\} \setminus \left\{\theta_2^k q, \sigma_{2,1}(k)>4, \sigma_{2,2}(k)>1\right\}  \right\rangle \cap \KK[x_1,x_2]\left[z^{(i,j)},\, i\leq 4, j\leq 1\right],
\end{equation}
and obtain a principal ideal whose defining polynomial is given by
\begin{dmath}\label{eq:elimIpq2}
\left(x_{1}^{12}+2 x_{1}^{11}-x_{1}^{10} x_{2}+x_{1}^{10}-2 x_{1}^{9} x_{2}-4 x_{1}^{7} x_{2}^{2}-5 x_{1}^{6} x_{2}^{2}-10 x_{1}^{4} x_{2}^{3}\right) z^{\left(1,1\right)}+\left(x_{1}^{11} x_{2}-3 x_{1}^{9} x_{2}+4 x_{1}^{8} x_{2}^{2}-2 x_{1}^{8} x_{2}+9 x_{1}^{7} x_{2}^{2}+10 x_{1}^{5} x_{2}^{3}-2 x_{1}^{5} x_{2}^{2}+30 x_{1}^{4} x_{2}^{3}+10 x_{1}^{3} x_{2}^{3}\right) z^{\left(2,1\right)}+\left(-2 x_{1}^{9} x_{2}^{2}-3 x_{1}^{8} x_{2}^{2}-3 x_{1}^{6} x_{2}^{3}+x_{1}^{6} x_{2}^{2}-12 x_{1}^{5} x_{2}^{3}-6 x_{1}^{4} x_{2}^{3}+12 x_{1}^{3} x_{2}^{4}+x_{1}^{2} x_{2}^{4}+2 x_{2}^{5}\right) z^{\left(3,1\right)}+\left(x_{1}^{7} x_{2}^{3}+2 x_{1}^{6} x_{2}^{3}+x_{1}^{5} x_{2}^{3}-2 x_{1}^{4} x_{2}^{4}-x_{1}^{3} x_{2}^{4}-2 x_{1} x_{2}^{5}\right) z^{\left(4,1\right)}.
\end{dmath}
The notation $\theta_2^{\leq n} p$ denotes all $\theta_2$-derivatives of $p$ of ($\theta_2$-)order at most $n$.

Hence for this second example, we obtained an ADE whose order (according to any of the definitions mentioned in \Cref{sec:derrule}) is not the sum of the orders of $p_1$ and $p_2$ from \eqref{eq:eg2p1p2}. \QEDA
\end{example}
 
To conclude that this is a general fact, we need to prove that the result of this example does not change if we consider higher $\theta_2$-order truncations of the corresponding differential ideal. Prior to that, we state a useful lemma.

\begin{lemma}[Increasing property]\label{lem:lem2} Let $y$ be a differential indeterminate. Consider the $\theta_2$-ranking in $R[y^{(\infty,\infty)}]$. For all non-negative integers $a_1,a_2,b_1,b_2$, and $k$, we have
\begin{equation}
    y^{(a_1,a_2)} < y^{(b_1,b_2)} \Longrightarrow \theta_2^k y^{(a_1,a_2)} < \theta_2^k y^{(b_1,b_2)}.
\end{equation}
In other words, the application of $\theta_2$ derivations preserves the ranking.
\end{lemma}
\begin{proof} Suppose $y^{(a_1,a_2)} < y^{(b_1,b_2)}$. By construction of the $\theta_2$-ranking, this is equivalent to
$$
 \begin{cases}
    a_1+a_2 < b_1 + b_2\\
    \text{or}\\
    a_2 < b_2
\end{cases}.
$$
We have $\theta_2^k y^{(a_1,a_2)} = y^{\sigma_2^{-1}(k) + (a_1,a_2)}$ and $\theta_2^k y^{(b_1,b_2)} = y^{\sigma_2^{-1}(k) + (b_1,b_2)}$. Therefore
$$
a_1+\sigma_{2,1}^{-1}(k)+a_2+\sigma_{2,2}^{-1}(k) < b_1 + \sigma_{2,1}^{-1}(k) + b_2 +\sigma_{2,2}^{-1}(k),\, \text{ and }\, a_2 +\sigma_{2,2}^{-1}(k) < b_2 + \sigma_{2,2}^{-1}(k).
$$
Hence $\theta_2^k y^{(a_1,a_2)} < \theta_2^k y^{(b_1,b_2)}.$\CQFD
\end{proof}

We can now state our theorem.

\begin{theorem}\label{theo:theo2} We use the same notations of this section. For all $N\in \NN$, $r\in\KK(x_1,x_2)(v_1,\ldots,v_N)$, the order of the least-order algebraic PDE fulfilled by the bivariate D-algebraic function given by $r(f_1(x_1,x_2),\ldots,f_N(x_1,x_2))$ is not bounded by the sum of the orders of their defining differential polynomials $p_1,\ldots,p_N$.
\end{theorem}
\begin{proof}
 We prove this by providing a counter-example, which is nothing else but the addition of zeros of $p_1$ and $p_2$ from \eqref{eq:eg2p1p2}. By \Cref{lem:lem2} it holds that the leading terms of the polynomials
 \begin{equation}
    \theta_2^{\infty} p_1,\, \theta_2^{\infty} p_2, \theta_2^{\infty} q \in \KK[x_1,x_2][y_1^{(\infty,\infty)},y_2^{(\infty,\infty)},z^{(\infty,\infty)}]
 \end{equation}
 are all distinct. Since those leading terms are linear over $\KK(x_1,x_2)[y_1^{(\infty,\infty)},y_2^{(\infty,\infty)},z^{(\infty,\infty)}]$, we deduce that each of their higher $\theta_2$-derivative linearly occurs starting from a specific $\theta_2$-truncation of the differential ideal 
 $$
 I_{p_1,p_2,q}\coloneqq\left\langle p_1, p_2, q \right\rangle^{(\infty,\infty)} \subset \KK[x_1,x_2][y_1^{(\infty,\infty)},y_2^{(\infty,\infty)},z^{(\infty,\infty)}].
 $$
 Therefore any (algebraic) elimination ideal computed from higher $\theta_2$-truncations of $I_{p_1,p_2,q}$ would either contain \eqref{eq:elimIpq2}, a $\theta_2$-derivative of it, or another differential polynomial of higher $\theta_2$-order; thus, supporting the fact that no higher $\theta_2$-order truncation of $I_{p_1,p_2,q}$ contains a differential polynomial in $R[z^{(\infty,\infty)}]$ of lower order than \eqref{eq:elimIpq2}. \CQFD
 
 % first Burchberger's criterion that 
 % \begin{equation}
 %     \left\langle \theta_2^{\infty} p_1,\, \theta_2^{\infty} p_2, \theta_2^{\infty} q \right\rangle
 % \end{equation}
 % is a Gr\"obner basis representation of 
 
 % if we use the lexicographical ordering where $z$ and its $\theta_2$-derivatives are ranked higher than the $y_i$'s (see (I.) and (II.) from \pageref{eq:derp2}). This implies that the elimination ideal from \eqref{eq:elimeg2} is also a Gr\"obner basis, 
\end{proof}

Note that the arguments used in the proof of \Cref{theo:theo2} are valid for any l.h.o. partial differential polynomials with initials in $\KK[x_1,x_2]$. Thus it implies that our algorithm always yields the least-order algebraic PDEs for the arithmetic of zeros of l.h.o. differential polynomials with initials in $\KK[x_1,x_2]$, and can therefore be used to define lower bounds for the order in those cases. It would be interesting to study the behavior of the algorithm for arbitrary algebraic PDEs with an analysis of the transcendence basis when the $\theta_2$-derivatives are taken. We think there is a connection between Kolchin polynomials and the growth of the number of partial derivatives involved in the truncated ideals (see \cite[Section 12]{kolchin1973differential}, \cite{levin2021bivariate}).

\subsection{Generalization and algorithm}

We now want to consider arbitrarily many independent variables. Suppose there are $l\geq 2$ independent variables. As we have seen in the previous two subsections, the definition of a total ordering encoded by the operator $\theta_l=\theta_{x_1,\ldots,x_l}$ using a bijective mapping from $\NN^l$ to $\NN$ is inherent to our algorithmic computations. One must always define how partial derivatives of differential indeterminates are selected because there is no natural way to do that like in the ordinary setting. A naive bijection can be defined in a recursive fashion with the Cantor pairing function in the base case. Indeed, one can show by induction that the $l$ maps defined by the recursion
\begin{equation}\label{eq:gammal}
\begin{split}
       &\gamma_j\colon \NN^j \longrightarrow \NN\\
       &\left(i_1,\ldots,i_j\right) \mapsto \sigma_2\left(\gamma_{j-1}(i_1,\ldots,i_{j-1}),i_j\right),
\end{split}
\end{equation}
$j=1,\ldots,l$, $\gamma_1=\text{Id}_{\NN}$, are bijective. However, these maps do not satisfy our desired properties. What is important for us in the choice of this bijection is the fulfillment of a property like \Cref{lem:lem2} that gives theoretical meanings to all outputs of our algorithm for solutions of l.h.o. algebraic PDEs whose initials are polynomials in $\KK[x_1,\ldots,x_l]$. For the $\gamma_j$ functions above, one verifies, for instance, that $(0,1,0)<(0,0,2)$ but $(1,0,0)+(0,1,0)=(1,1,0)>(1,0,2)=(1,0,0)+(0,0,2)$, which is not desired.

Having said that, there could still be several possible choices. For this paper, the choice is the Cantor $l$-tuple function which generalizes our chosen $\sigma_2$ in $l$ dimensions. This corresponds to the inverse lexicographic ordering, often called co-lexicographic ordering. For monomial ordering, this is understood as the degree lexicographic ordering with the variable ranked in the inverse order. Precisely, we consider $\sigma_l$ as the $l$-tuple function that ranks $l$-tuples of integers as follows:
\begin{equation}\label{eq:thetalrank}
    (a_1,\ldots,a_l) < (b_1,\ldots,b_l)\, \iff \, \begin{cases}
        a_1+\cdots+a_l < b_1+\cdots+b_l\\[0.33mm]
        \text{or}\\[0.33mm]
        \begin{cases}
          a_1+\cdots+a_l = b_1+\cdots+b_l,\, \text{ and}\\
          \exists\, i, j,\, 1\leq j \leq l\,, 1\leq i \leq l-j,\, \colon \, a_{j+i}=b_{j+i}, a_j<b_j
        \end{cases}
    \end{cases}.
\end{equation} 
We then take $\theta_l$ as the derivation rule that uses $\sigma_l$ and its reciprocal $\sigma_l^{-1}$ similarly as $\theta_2$ uses $\sigma_2$ and its reciprocal $\sigma_2^{-1}$. Note that an explicit algebraic formula may be derived for this Cantor $l$-tuple function. We will not dive into such details; the combinatorial investigation can be found in  \cite{ruskey2003combinatorial}. We only require fulfillment of properties like \Cref{prop:prop2} and \Cref{lem:lem2}. The increasing property generalizes to arbitrarily many independent variables as we show in the next proposition.

\begin{proposition}[Increasing property continued] Let $y$ be a differential indeterminate. Consider the $\theta_l$-ranking in $R[y^{(\infty,\overset{l}{\ldots},\infty)}]$. For all non-negative integer-valued tuples $a,b\in\NN^l$ and an integer $k$, we have
\begin{equation}
    y^a < y^b \Longrightarrow \theta_l^k y^a < \theta_l^k y^b.
\end{equation}
\end{proposition}
\begin{proof} Let $a,b\in\NN^l$ such that $y^a<y^b$ w.r.t. the $\theta_l$-ranking. By definition $\theta_l^k y^a = y^{\sigma_l^{-1}(k)+a}$. Thus the application of $\theta_l^k$ to $y^a$ and $y^b$ adds $\sigma_{l,i}^{-1}(k)\geq 0$ to each $a_i$ and $b_i$, $i=1,\ldots,k$. This preserves the inequality by the definition of the ranking in \eqref{eq:thetalrank}.\CQFD
\end{proof}
One easily verifies that a similar property as \Cref{prop:prop2} also holds. Hence using $\theta_l$ we can compute ADEs for arithmetic of multivariate D-algebraic functions. Our algorithm that summarizes all the necessary steps is given by \Cref{algo:Algo2}.

%\vspace*{-.2cm}

\begin{algorithm}[!ht]\caption{Searching ADEs fulfilled by rational expressions of multivariate D-algebraic functions}\label{algo:Algo2}
    \begin{algorithmic} 
    \Require A rational function $r\in\KK(\mathbf{x},y_1,\ldots,y_N)$, and $N$ ADEs associated to the differential polynomials $p_i\in \KK(x_1,\ldots,x_l)[y_i^{(\infty,\overset{l}{\ldots},\infty)}]$ all of derivations $\Delta_2=\{\partial_{x_i},\, i=1,\ldots,l\}$ (or simply $\theta_{x_1,\ldots,x_l}\coloneqq\theta_l$). 
    \Ensure $0$ (unsuccessful situation) or an ADE fulfilled by $r(f_1,\ldots,f_N)$, where $f_i$ is a zero of $p_i$, $i=1,\ldots,N$. By default, for each independent variable, the maximum order of the ADE sought is the sum of the orders w.r.t. that variable. Optionally, one can choose $(\nu_1,\ldots,\nu_l)$ as a maximum order and make the algorithm starts from step \ref{stepstart}. If a desired ADE exists within the prescribed order bound, then it will be found in the conditions of \Cref{theo:theo3}; otherwise, a higher order bound is required.
    \begin{enumerate}
        \item For each differential polynomial $p_i$, let $n_{i,j},\, j=1,\ldots,l$ be the orders of $p_i$ w.r.t. $x_j$.
        \item Set $\mathcal{N}\coloneqq (\nu_1,\ldots,\nu_l) = \left(\sum_{i=1}^N n_{i,1},\ldots, \sum_{i=1}^N n_{i,l}\right)$.
        \item Set $\nu=\sigma_l(\nu_1,\ldots,\nu_l)$. //\textit{Comment:} $(\nu_1,\ldots,\nu_l)$ can be optionally given in the input.
        \item\label{stepstart} Let $\mathcal{R}$ be the numerator of the rational expression $w-r(y_1,\ldots,y_N)$.
        \item Let $E$ be the set
        $$
          \left\{ \theta_l^k \mathcal{R}: \sigma_{l,1}^{-1}(k)\leq \nu_1,\ldots, \sigma_{l,l}^{-1}(k)\leq \nu_l,\,  k=0,\ldots,\nu \right\}
        $$
        \item Let $(m_1,\ldots, m_l)$ be the minimum order (w.r.t. the $\theta_l$-ranking) that appears among the orders $(n_{i,1},\ldots,n_{i,l}), i=1,\ldots,N$.
        \item Let $d=\sigma_l(\nu_1-m_1,\ldots,\nu_l-m_l)$.
        
        \item Update the set $E$ by adding the polynomials $p_1,\ldots,p_N$, and their first $d$ $\theta_l$-derivatives to it,i.e.,
        $$
            E \coloneqq E \cup \left\{ \theta_l^k p_i, k=1,\ldots,d, i=1,\ldots,N \right\}
        $$
        
        \item Let $D_i\in\NN^l$, be the maximum order w.r.t. the $\theta_l$-ranking of the differential indeterminate $y_i$ in $E$, $i=1,\ldots,N$.
        
        \item Let
        $$
            I\coloneqq \langle E \rangle \subset \KK[x_1,\ldots,x_l][y_1^{\leq D_1},\ldots,y_N^{\leq D_N},z^{(\leq \nu_1,\ldots \leq \nu_l)}],
        $$
        be the ideal defined by the polynomials in $E$.
        \item Compute the elimination ideal
        $$
            \mathcal{I} \coloneqq I \cap \KK[x_1,\ldots,x_l][z^{(\leq \nu_1,\ldots \leq \nu_l)}],
        $$
        using the lexicographic monomial ordering as defined in (I.) and (II.) from page \pageref{pagerank}. 

\end{enumerate}
 \algstore{pause1}
 \end{algorithmic}
 \end{algorithm}

  \clearpage 

 \begin{algorithm}
%    \ContinuedFloat
     \caption{Searching ADEs fulfilled by rational expressions of multivariate D-algebraic functions}
     \begin{algorithmic}
         \algrestore{pause1}	
         \State
    \begin{enumerate}
         \setcounter{enumi}{11}

        \item If $\mathcal{I}\neq \langle 0 \rangle$ then:
            \begin{itemize}
                \item Let $p$ be the polynomial of lowest degree among those of lowest $\theta_l$-order in $\mathcal{I}$.
                \item \textbf{Stop and return} the writing of $p=0$ in terms of partial derivatives.
            \end{itemize}
        \item ($\mathcal{I} = \langle 0 \rangle$) While $ d \leq \nu$ do:
            \begin{itemize}
                \item $d \coloneqq d+1$ (update $d$ by $d+1$)
                %\item Let $(d_1,\ldots,d_l) = \sigma_l^{-1}(d)$
                \item Repeat steps 8 to 12
            \end{itemize}
        \item \textbf{Return} $0$ (meaning that ``No ADE of order component-wisely less than $(\nu_1,\ldots,\nu_l)$ found'').
    \end{enumerate}
    \end{algorithmic}
\end{algorithm}

\vspace*{-.5cm}

\begin{theorem}\label{theo:theo3} \Cref{algo:Algo2} is correct, and when successful (non-zero output) with input ADEs that are l.h.o. algebraic PDEs with initials in $\KK[x_1,\ldots,x_l]$, it returns an algebraic PDE of least-order w.r.t. the total ordering induced by $\theta_l$.
\end{theorem}
\begin{proof} As we have seen in the previous section, \Cref{algo:Algo2} is an exhaustion of the possibilities to find a least-order ADE with a given bound. Therefore correctness easily follows. However, when the algorithm halts, there is no evidence that the output is necessarily of minimal order. In the case of l.h.o. algebraic PDEs whose initials are polynomials in the independent variables, the justification is the same as for the counter-example given in the proof of \Cref{theo:theo2}: the first non-trivial elimination always yields the least-order algebraic PDE possible.\CQFD
\end{proof}

\begin{remark} We may speed up \Cref{algo:Algo2} by checking the algebraic dependency with a Jacobian condition (see \cite[Theorem 2.2]{ehrenborg1993apolarity}) instead of always computing Gr\"obner bases.
\end{remark}

Let us say a few words on the reason why the triangular set arguments (see \Cref{lem:lem1}) do not apply the same way to guarantee the minimality of the order for arbitrary algebraic PDEs like in the ordinary case. When applying the $\theta_l$-derivations, the transcendence degree increases because some variables occur independently of the algebraic dependencies encoded by a single PDE. This may explain why partial differential equations are often studied as systems instead of single equations like in the ordinary setting. To illustrate this phenomenon, consider our example of bivariate differential polynomials $p_1=\theta_2^2y_1 + x_2\theta_2^4y_1$, $p_2=x_1\theta_2^1y_2-\theta_2^3y_2$, for which we were looking for an ADE fulfilled by sums of their zeros. Using the $\theta_2$-ranking, the following system can be used to portray how the algebraic dependencies occur.
 \begin{equation}\label{eq:sysmultvars}
     \begin{cases}
         \theta_2^1 y_1 = w_1 \\
         \theta_2^2 y_1 = w_2\\
         \theta_2^3 y_1 = w_3\\
         \theta_2^4 y_1 =-\frac{w_2}{x_2}\\
         \theta_2^1 y_2 = w_4\\
         \theta_2^2 y_2 = w_5\\         \theta_2^3 y_2 = x_1\, w_4\\
                   z    = y_1+y_2
    \end{cases}.
 \end{equation}
 Since these are linear equations, substitution easily does the elimination. The main question is to know whether we can always substitute higher $\theta_2$-derivatives of $y_1$ and $y_2$ in higher $\theta_2$-derivatives of $z$ by some rational expression in $w_1,\ldots,w_5$. The answer is no since applying $\theta_2^4$ to the last equation yields
 $$
 \theta_2^4 z = \theta_2^4 y_1 + \theta_2^4 y_2 = -\frac{w_2}{x_2} + \theta_2^4 y_2,
 $$
and there is no link to $\theta_2^4 y_2$. Thus $\theta_2^4 y_2$ could be seen as a new transcendent element to consider in the elimination process. Such a situation never happens with ordinary differential equations because $\theta_1^{k+1}y=\theta_1 (\theta_1^k y)$. Note that an algebraic dependency between $z$-derivatives will certainly arise at some point; however, it does not seem obvious to predict when that will occur, like in the ODE case. 

\subsection{Implementation}

Our package \texttt{NLDE} contains a subpackage called \texttt{MultiDalg} that is being designed to perform operations with multivariate D-algebraic functions. The current implementation of \Cref{algo:Algo2} is \texttt{arithmeticMDalg} in \texttt{MultiDalg} for both unary and $N$-ary arithmetic of multivariate D-algebraic functions. The syntax is similar to that of \texttt{NLDE:-arithmeticDalg}. The monomial ordering used is either the lexicographic (default) or the lexdeg ordering with the corresponding $\theta_l$-ranking, where $l$ is the number of independent variables.

\begin{example} Consider the three PDEs:
\begin{equation}\label{eq:eq13}
    \frac{\partial f}{\partial x} = \frac{\partial f}{\partial y},\, \frac{\partial f}{\partial y} = \frac{\partial f}{\partial z}, \, \frac{\partial f}{\partial z} = \frac{\partial f}{\partial x}.
\end{equation}
We view them as three independent PDEs and would like to compute a fourth PDE for the sum of their solutions. Using our package, this can be done with the code below.
\begin{lstlisting}
> with(NLDE:-MultiDalg): #to load the subpackage.
> ADE1:=diff(U(x,y,z),x)=diff(U(x,y,z),y):
> ADE2:=diff(V(x,y,z),y)=diff(V(x,y,z),z):
> ADE3:=diff(W(x,y,z),z)=diff(W(x,y,z),x):
> arithmeticMDalg([ADE1,ADE2,ADE3],[U(x,y,z),V(x,y,z),W(x,y,z)],T=U+V+W)
\end{lstlisting}
\begin{dmath}\label{eq:eg1MultiDalg}
\frac{\partial^{3}}{\partial x^{2}\partial y}T \! \left(x ,y ,z \right)-\frac{\partial^{3}}{\partial x \partial y^{2}}T \! \left(x ,y ,z \right)-\frac{\partial^{3}}{\partial x^{2}\partial z}T \! \left(x ,y ,z \right)+\frac{\partial^{3}}{\partial y^{2}\partial z}T \! \left(x ,y ,z \right)+\frac{\partial^{3}}{\partial x \partial z^{2}}T \! \left(x ,y ,z \right)-\frac{\partial^{3}}{\partial y \partial z^{2}}T \! \left(x ,y ,z \right)=0
\end{dmath}
One can use Maple's \texttt{pdsolve} command to check that sum of solutions to the PDEs in \eqref{eq:eq13} are solutions of \eqref{eq:eg1MultiDalg} and that the solutions of \eqref{eq:eg1MultiDalg} are sums of solutions to the PDEs in \eqref{eq:eq13}. A nice algebraic perspective for solving such linear PDEs is given in \cite[Section 3.3]{michalek2021invitation}. \QEDA
\end{example}

\begin{example} We show how to recover the results of our explanatory examples with our implementation. To get \eqref{eq:addeq1}, one uses the code below. We do not display the output as it is identical to \eqref{eq:addeq1}. The last line gives the result.
\begin{lstlisting}
> ADE1:=diff(y[1](x[1],x[2]),x[1],x[2])*x[2]+diff(y[1](x[1],x[2]),x[2])=0:
> ADE2:=diff(y[2](x[1],x[2]),x[1])*x[1]-diff(y[2](x[1],x[2]),x[1],x[1])=0:
> arithmeticMDalg([ADE1,ADE2],[y[1](x[1],x[2]),y[2](x[1],x[2])],z=y[1]+y[2]):
\end{lstlisting}
 For the second example, to check that there is no ADE of order less than $(3,1)$, component wise, one uses the following code.
\begin{lstlisting}
> ADE1:=diff(y[1](x[1],x[2]),x[1],x[2])*x[2]+diff(y[1](x[1],x[2]),x[2])*x[1]=0:
> ADE2:=x[1]^2*diff(y[2](x[1],x[2]),x[1])-diff(y[2](x[1],x[2]),x[1],x[1])*x[2]=0:
> arithmeticMDalg([ADE1,ADE2],[y[1](x[1],x[2]),y[2](x[1],x[2])],z=y[1]+y[2])
\end{lstlisting}
\begin{dmath}
    0.
\end{dmath}
The $0$ in the output should be interpreted as ``no ADE of order less than the sum of the orders was found''. To get the ADE associated to \eqref{eq:elimIpq2}, one supplies the optional argument \texttt{maxord=[4,1]} and gets:
\begin{lstlisting}
> arithmeticMDalg([DE1,DE2],[y[1](x[1],x[2]),y[2](x[1],x[2])],
                                z=y[1]+y[2],maxord=[4,1])
\end{lstlisting}
\begin{dmath}\label{eq:eg2MultiDalg}
\left(x_{1}^{12}+2 x_{1}^{11}-x_{1}^{10} x_{2}+x_{1}^{10}-2 x_{1}^{9} x_{2}-4 x_{1}^{7} x_{2}^{2}-5 x_{1}^{6} x_{2}^{2}-10 x_{1}^{4} x_{2}^{3}\right) \left(\frac{\partial^{2}}{\partial x_{1}\partial x_{2}}z \! \left(x_{1},x_{2}\right)\right)+\left(x_{1}^{11} x_{2}-3 x_{1}^{9} x_{2}+4 x_{1}^{8} x_{2}^{2}-2 x_{1}^{8} x_{2}+9 x_{1}^{7} x_{2}^{2}+10 x_{1}^{5} x_{2}^{3}-2 x_{1}^{5} x_{2}^{2}+30 x_{1}^{4} x_{2}^{3}+10 x_{1}^{3} x_{2}^{3}\right) \left(\frac{\partial^{3}}{\partial x_{1}^{2}\partial x_{2}}z \! \left(x_{1},x_{2}\right)\right)+\left(-2 x_{1}^{9} x_{2}^{2}-3 x_{1}^{8} x_{2}^{2}-3 x_{1}^{6} x_{2}^{3}+x_{1}^{6} x_{2}^{2}-12 x_{1}^{5} x_{2}^{3}-6 x_{1}^{4} x_{2}^{3}+12 x_{1}^{3} x_{2}^{4}+x_{1}^{2} x_{2}^{4}+2 x_{2}^{5}\right) \left(\frac{\partial^{4}}{\partial x_{1}^{3}\partial x_{2}}z \! \left(x_{1},x_{2}\right)\right)+\left(x_{1}^{7} x_{2}^{3}+2 x_{1}^{6} x_{2}^{3}+x_{1}^{5} x_{2}^{3}-2 x_{1}^{4} x_{2}^{4}-x_{1}^{3} x_{2}^{4}-2 x_{1} x_{2}^{5}\right) \left(\frac{\partial^{5}}{\partial x_{1}^{4}\partial x_{2}}z \! \left(x_{1},x_{2}\right)\right)=0.
\end{dmath}\QEDA
\end{example}
\section{Discussions: applications and future developments}\label{sec:disc}
\subsection{Applications}\label{sec:applications}

Let us discuss the potential application of our results, and in particular, the use of our software. First of all, it is worthwhile to remember that these elimination techniques are familiar to many scientists from many disciplines. However, it is often the case that the software they use, when the computations involved are tedious to do by hand, is tailored to a very specific problem and cannot be easily adapted when the underlying (D-algebraic) functions are changed. Usually one can notice that their accompanying codes for related problems are adaptations of an old code that was designed for some popular D-algebraic functions. See for instance, the recurrent use of elimination to find the fourth Painlev\'{e} transcendent of Hamiltonian representations in \cite{dzhamay2021different}. Since several functions in the sciences such as Physics, Biology, or Statistics are defined by algebraic ordinary or partial differential equations, usually with some initial values, one can use our software to compute differential equations satisfied by rational expressions of solutions to given ADEs. Let us demonstrate this with some concrete examples.
\begin{example} In \cite[Problem 6.24]{teschl2012ordinary}, the following algebraic ODE is derived from the Korteweg-de Vries equation:
\begin{equation}\label{eq:appl1}
    - c\, v'(x) + v^{(3)}(x)+ 6\,v(x)\,v'(x) = 0.
\end{equation}
Physicists are interested in solutions which satisfy $\lim_{x\longrightarrow \pm \infty} v(x) =0$. This limits the admissible parameter $c$ on the shape of $v(x)$. One looks for a linear transformation of $v$ that eliminates the term $-c\,v'$ from \eqref{eq:appl1}. Using \texttt{NLDE}, we compute an ADE for any such linear transformation and deduce a desired equation.
\begin{lstlisting}
> ADEKV:=-c*diff(v(x),x)+diff(v(x),x,x,x)+6*v(x)*diff(v(x),x)=0:
> NLDE:-unaryDalg(ADEKV,v(x),w=C[1]*v+C[2])
\end{lstlisting}

\vspace{-0.3cm}

\begin{dmath}\label{eq:appl1C1C2}
6 \, w \! \left(x \right) \left(\frac{d}{d x}w \! \left(x \right)\right)+\left(-c\,C_{1}-6\,C_{2}\right) \left(\frac{d}{d x}w \! \left(x \right)\right)+C_{1} \left(\frac{d^{3}}{d x^{3}}w \! \left(x \right)\right)=0.
\end{dmath}
From this output, one deduces that for $C_2=-(c\,C_1)/6$ the transformation $w=C_1\,v+C_2$ eliminates the first derivative in the resulting equation. Thus we have found infinitely many transformations to eliminate the first derivative. A simple one would be $w=-v+c/6$, which yields
\begin{lstlisting}
> NLDE:-unaryDalg(ADEKV,v(x),w=-v+c/6)
\end{lstlisting}

\vspace{-0.3cm}

\begin{dmath}\label{eq:appl1res}
6 w \! \left(x \right) \left(\frac{d}{d x}w \! \left(x \right)\right)-\frac{d^{3}}{d x^{3}}w \! \left(x \right)=0.
\end{dmath}
Note that we recover the same result of the book by using the linear transformation $w=(-6\,v+c)/12$ (given in the book as $v=-2\,w+c/6)$. This gives the ADE $w^{(3)}=12\,w'\,w$. \QEDA
\end{example}

\begin{example} Observe that \eqref{eq:appl1} can be integrated once to get
\begin{equation}\label{eq:appl2}
    v'' - c\,v + 3\, v^2 + a = 0,
\end{equation}
where $a$ is the constant of integration. Note that the shape of $v$ also depends on $a$. In \cite[Problem 6.24]{teschl2012ordinary}, it is mentioned that $v(x)=-2\wp(x)+c/6$ satisfies \eqref{eq:appl2}, where $\wp(x)$ is the Weierstrass elliptic function. Let us verify this fact. We use our software to get an ADE fulfilled by the rational expression $-2\wp(x)+c/6$ from the ADE of the Weierstrass elliptic function.
\begin{lstlisting}
> ADEwp:=diff(p(x),x)^2=4*p(x)^3-g[2]*p(x)-g[3]: #Weierstrass equation
> ADE1:=NLDE:-unaryDalg(ADEwp,p(x),v=-2*p+c/6)
\end{lstlisting}

\vspace{-0.4cm}

\begin{dmath*}
\mathit{ADE1} \coloneqq 216 v \! \left(x \right)^{3}-108 v \! \left(x \right)^{2} c +\left(18 c^{2}-216 g_{2} \right) v \! \left(x \right)+108 \left(\frac{d}{d x}v \! \left(x \right)\right)^{2}-c^{3}+36 c g_{2} +432 g_{3} =0.
\end{dmath*}
Since \eqref{eq:appl2} is a second-order ADE, we take one derivative of the result and factor the output.
\begin{lstlisting}
> factor(diff(ADE1,x))
\end{lstlisting}
\begin{dmath}\label{eq:appl22}
18 \left(\frac{d}{d x}v \! \left(x \right)\right) \left(36 v \! \left(x \right)^{2}-12 c\,v \! \left(x \right)+c^{2}+12 \frac{d^{2}}{d x^{2}}v \! \left(x \right)-12 g_{2} \right)=0.
\end{dmath}
Neglecting the constant solutions, we are interested in the solutions of the ADE
\begin{equation}
   \frac{d^{2}}{d x^{2}}v \! \left(x \right) - c\,v \! \left(x \right) +  3 v \! \left(x \right)^{2} + \left(\frac{c^{2}}{12} - g_{2}\right)=0.
\end{equation}
Therefore by taking the constant of integration $a$ as $\left(\frac{c^{2}}{12} - g_{2}\right)$, we deduce that the claim holds.
\QEDA
\end{example}
The above two examples present how one can apply our result to ``basic'' (symbolic) calculations in Physics. Note that the Korteweg-de Vries equation is an algebraic PDE that models shallow water waves. One of its outstanding features is the existence of so-called \textit{soliton}, i.e., waves that travel in time without changing their shape (see \cite[Problem 6.24]{teschl2012ordinary}).

For algebraic PDEs, there seems to be more interest in numerical methods than symbolic ones. The usual non-expressiveness of general solutions of PDEs may explain this. Nevertheless, as for algebraic ODEs, our software can compute ADEs satisfied by rational expressions of solutions to algebraic PDEs (see \Cref{eg:bonus}). However, note that so far, we only considered additions of solutions of linear algebraic PDEs. The main reason behind that is because finding least-order ADEs for addition of solutions to linear ADEs can be performed with Gauss elimination, and so carrying out the same computations using Gr\"obner bases is essentially the same. When it comes to other operations like the product, Gauss elimination will not generally give the desired least-order ADE, and then the computations with Gr\"obner bases pay the price of dealing with higher degree equations. For multivariate D-algebraic functions, this price is notoriously visible by the high CPU time encountered when running our implementation for product and division using the pure lexicographic monomial ordering. This behavior could be predicted by the number of variables involved in the elimination steps of \Cref{algo:Algo2}.  We recommend the specification \texttt{ordering=lexdeg} which often (not always) enables computations in shorter timings.

There are many situations in which the algorithm proves useful. Let us give an idea of applications of \Cref{algo:Algo2} related to solving PDEs. We construct PDEs out of solutions to known ODEs. This may be seen as a \textit{reverse engineering} approach of the method of separating variables. The latter consists of finding solutions that can be expressed rationally in terms of solutions of ODEs in each independent variable. For instance, given a PDE in the independent variables $x_1$ and $x_2$, we commonly look for solutions of the product form $F_1(x_1)\,F_2(x_2)$, where $F_1$ and $F_2$ are solutions of some ODEs in $x_1$ and $x_2$, respectively. The product form is not very interesting for linear ODEs with constant coefficients, because viewing $F(x_i)$ as a constant for a linear ODE in $x_j, i\neq j$, implies already that $F(x_i)\,F(x_j)$ is a solution of that ODE. Therefore the algorithm may just return a derivative of that ODE. We would like to consider more tricky examples.
\begin{example} Suppose we want to construct an algebraic PDE whose solutions are of the form
\begin{equation}\label{eq:shapeLogis}
    \frac{1}{1+\exp(a\,x_1)\,\exp(b\,x_2)} = \frac{1}{1+\exp\left(a\, x_1 + b\, x_2\right)},
\end{equation}
with arbitrary constants $a,b$. Note that this is a well-known form in Statistics for the \textit{logistic distribution} (see for instance \cite{logisticDbalakrishnan}), where one of the variables would be interpreted as the \textit{mean}. 

We use the linear ADEs $y_1'-a\,y_1=0$ for $\exp(a\,x_1)$ and $y_2'-b\,y_2=0$ for $\exp(b\,x_2)$. Observe that these ADEs could also be taken as PDEs, i.e, $y_1^{(1,0)} - a\,y_1=0$ and $y_2^{(0,1)} - b\,y_2 =0$; however, using these PDEs would lead to an algebraic PDE whose solutions have a more general form. For instance, a solution of $y^{(1,0)} - a y=0$ is of the form $F(x_2)\, \exp(a\,x_1)$, which is not desired for our target form \eqref{eq:shapeLogis}.

We here reveal a feature of our implementation that enables algebraic ODE inputs. The computations still use $\theta_l$ derivations as described in \Cref{algo:Algo2}. This decreases the number of variables used in the elimination step and makes the code get the result more efficiently. Indeed, when we apply $\theta_2$ derivations, $y_1$ and its derivatives w.r.t. $x_1$ are treated as constants for derivation w.r.t. $x_2$, while $y_2$ and its derivatives w.r.t. $x_2$ are treated as constants for derivation w.r.t. $x_1$. This results in the cancellation of all variables $y_1^{(i, j)}$ and $y_2^{(j, i)}$, $i, j\in \NN$, with $j\neq 0$. Let us now compute the desired algebraic PDE.

\begin{lstlisting}
> ADE1:=-a*y[1](x[1]) + diff(y[1](x[1]), x[1]) = 0:
> ADE2:=-b*y[2](x[2]) + diff(y[2](x[2]), x[2]) = 0:
> arithmeticMDalg([ADE1,ADE2],[y[1](x[1]),y[2](x[2])],
                       z=1/(1+y[1]*y[2]),[x[1],x[2]])
\end{lstlisting}

\vspace{-0.3cm}

\begin{dmath}\label{eq:lasteg1}
z \! \left(x_{1},x_{2}\right)^{2} b -b z \! \left(x_{1},x_{2}\right)-\frac{\partial}{\partial x_{2}}z \! \left(x_{1},x_{2}\right)=0.
\end{dmath}
The obtained equation is the logistic differential equation, which could also be seen as an ODE in $x_2$. This corresponds to solutions of the form
\begin{equation}\label{eq:sollasteg1}
z \! \left(x_{1},x_{2}\right)=\frac{1}{1+{\mathrm e}^{b x_{2}} \textit{F} \! \left(x_{1}\right)},
\end{equation}
where $F$ is an arbitrary function in $x_1$. This is a more restricted form compared to the result obtained with PDEs as inputs. \QEDA
\end{example}

\begin{example} Taking the same algebraic ODEs as inputs, we now look for a PDE whose solutions have the form
\begin{equation}\label{eq:form2}
    \frac{\exp(a\,x_1)}{1+\exp(a\,x_1+b\,x_2)}.
\end{equation}
Our package yields:
\begin{lstlisting}
> arithmeticMDalg([ADE1,ADE2],[y[1](x[1]),y[2](x[2])],
                       z=y[1]/(1+y[1]*y[2]),[x[1],x[2]])
\end{lstlisting}

\vspace{-0.3cm}

\begin{dmath}\label{eq:eqform2}
a b z \! \left(x_{1},x_{2}\right)+\left(\frac{\partial}{\partial x_{2}}z \! \left(x_{1},x_{2}\right)\right) a -\left(\frac{\partial}{\partial x_{1}}z \! \left(x_{1},x_{2}\right)\right) b =0.
\end{dmath}
In this case, Maple \texttt{pdsolve} command finds the solution:
\begin{dmath}\label{eq:solform2}
z \! \left(x_{1},x_{2}\right)=\textit{F} \! \left(\frac{x_{1} a +b x_{2}}{b}\right) {\mathrm e}^{x_{1} a},
\end{dmath}
for an arbitrary function $F$. It is obvious that \eqref{eq:form2} is of this form. The obtained equation is a linear PDE that can also be used for the symbolic integration of \eqref{eq:form2} (see \cite{zeilberger1991method}, \cite{koutschan2009advanced}).\QEDA
\end{example}

Remark that the calling syntax of \texttt{arithmeticMDalg} in the above two examples contains a bracket list with all the independent variables, here $x_1,x_2$, given as the last argument. This is how the code is called when there are algebraic ODEs in the input.

\begin{example}[Bonus example]\label{eg:bonus} We give a PDE fulfilled by the square of the probability density function
$$f(x,\mu) \coloneqq \frac{1}{\sigma \sqrt{2\,\pi}}\exp\left(-\frac{1}{2}\left(\frac{x-\mu}{\sigma}\right)^2\right),$$ 
of the univariate normal distribution $\mathcal{N}(\mu,\sigma^2)$, seen as a bivariate function in the mean and the indeterminate variable. We use \texttt{FPS:-HolonomicPDE} (see \cite{FPS}) to derive a partial PDE w.r.t. each independent variable, and use \texttt{arithmeticMDalg} to find an algebraic PDE for the product of their solutions.
\begin{lstlisting}
> f:=exp(-((x-mu)/sigma)^2/2)/(sigma*sqrt(2*Pi)):
> PDE1:=FPS:-HolonomicPDE(f,u(x,mu),partialwrt=x)
\end{lstlisting}
\begin{equation}\label{eq:pde1Normal}
\mathit{PDE1} \coloneqq \sigma^{2} \left(\frac{\partial}{\partial x}u \! \left(x ,\mu \right)\right)-\left(-x +\mu \right) u \! \left(x ,\mu \right)=0
\end{equation}

\vspace{-0.12cm}

\begin{lstlisting}
> PDE2:=FPS:-HolonomicPDE(f,v(x,mu),partialwrt=mu)
\end{lstlisting}
\begin{equation}\label{eq:pde2Normal}
\mathit{PDE2} \coloneqq \sigma^{2} \left(\frac{\partial}{\partial \mu}v \! \left(x ,\mu \right)\right)-\left(x -\mu \right) v \! \left(x ,\mu \right)=0
\end{equation}

\vspace{-0.12cm}

\begin{lstlisting}
> arithmeticMDalg([PDE1,PDE2],[u(x,mu),v(x,mu)],z=u*v)
\end{lstlisting}
\begin{equation}\label{eq:pdeNormalsquare}
-z \! \left(x ,\mu \right) \left(\frac{\partial^{2}}{\partial \mu \partial x}z \! \left(x ,\mu \right)\right) \sigma^{2}+\left(\frac{\partial}{\partial \mu}z \! \left(x ,\mu \right)\right) \left(\frac{\partial}{\partial x}z \! \left(x ,\mu \right)\right) \sigma^{2}+2 z \! \left(x ,\mu \right)^{2}=0.
\end{equation}
Solutions to this PDE have the form
\begin{equation}
    F_2 \! \left(x \right) \exp\left(F_1 \left(\mu \right)+\frac{2 x \mu}{\sigma^{2}}\right),
\end{equation}
with arbitrary functions $F_1$ and $F_2$.
%The obtained PDE may be investigated to find an ODE fulfilled by the volume (a univariate function in $\mu$) of the revolution of $f$ around the $x$-axis.
\QEDA
\end{example}

\subsection{Conclusion and future work}

We completed the exposition of \cite{RSB2023} on univariate D-algebraic functions by providing an algorithm (see \Cref{algo:Algo1}) to compute least-order ADEs for rational expressions of solutions to ADEs that are not linear in their highest occurring derivatives. This was accompanied by \Cref{lem:lem1}, which established its correctness. We have also proposed an algorithm for doing these computations with multivariate D-algebraic functions using an ordering on the semigroup of derivations that allows us to keep track of the application of partial derivatives (see \Cref{theo:theo3} and \Cref{algo:Algo2}). From that, we established \Cref{theo:theo2}, which states that the sum of orders of given algebraic PDEs is not always a bound for the order of an algebraic PDE fulfilled by rational expressions of solutions to those ADEs. We implemented the given algorithm using the computer algebra system Maple. The accompanying software can be downloaded from \href{https://mathrepo.mis.mpg.de/OperationsForDAlgebraicFunctions/index.html}{MathRepo NLDE} or \cite{NLDE} (with the source code) as a subpackage of the \texttt{NLDE} package. In certain calculations, our software may find similar results with the known packages \cite{gfun}, \cite{Mgfun}, \cite{kauers2019multivariate}; however, we are here in a more general setting as we mentioned earlier. In the univariate setting, \Cref{algo:Algo1} compares better with \cite{bachler2012algorithmic,boulier1995representation}. \Cref{eg:eg2} illustrates one such example where our algorithm outperforms those existing methods. We developed our approach from the works in \cite{hong2020global,dong2023differential}. Our software can be used in many scientific disciplines where eliminations of variables in partial or ordinary differential equations are of interest. We have given some in \Cref{sec:applications}. A Maple worksheet with the examples of the article can be downloaded from \href{https://mathrepo.mis.mpg.de/ArithmeticOfDAlgebraicFunctions}{Arithmetic of D-Algebraic functions} or \cite{NLDE}.

We remark from references on partial differential equations that the compositions of solutions of algebraic PDEs with solutions of algebraic ODEs is often used to simplify algebraic PDEs. For instance, to get the equation in \eqref{eq:appl1}, one considers the composition $u(x,t)=v(x-c\,t)$ with the Korteweg-de Vries equation
\begin{equation}
    \frac{\partial}{\partial t} u(t,x) + \frac{\partial^3}{\partial x^3} u(t,x) + 6\,u(t,x)\,\frac{\partial }{\partial x} u(t,x) = 0.
\end{equation}
We think that an adaptation of \Cref{algo:Algo2} for compositions can be used to make such transformations of algebraic PDEs into algebraic ODEs more systematic.

We also highlight the potential use of $\theta_l$-derivations for symbolic integration, which is also concerned with elimination. Indeed, given a symbolic expression in several variables $f(x_1,\ldots,x_n)$, it is possible to develop a Fasenmyer-like algorithm (see \cite[Theorem 11.4]{koepf2021computer}) that looks for linear algebraic PDEs with eliminated variables $S\subset \{x_1,\ldots,x_n\}$ among the coefficients, such that a differential equation for the integral $\int_S f(x_1,\ldots,x_n)$ can be deduced as a generalization of Feynman’s method. A preliminary implementation of this technique is provided in the \texttt{FPS} package by the \texttt{HolonomicPDE} command  (see \cite{FPS}). For a general knowledge of symbolic integration, we refer the reader to the non-exhaustive list of references \cite{zeilberger1991method,chyzak2000extension,kauers2011,koutschan2009advanced,chen2017some}.

\medskip

\textbf{Acknowledgment.} The author thanks Bernd Sturmfels for encouraging this work. He appreciates Gleb Pogudin for reading and commenting on earlier versions of the article, and the anonymous referees for their valuable comments. He also thanks Rida Ait El Manssour and Anna-Laura Sattelberger for useful discussions.

\end{document}